\documentclass[a4paper, 12pt]{article}
\usepackage{amsmath}
\usepackage{amsfonts}
\usepackage{amssymb}
\usepackage{subcaption}
\usepackage{enumitem}
\usepackage{graphicx}
\usepackage[a4paper, left=3cm, right=3cm, top=2.54cm, bottom=2.54cm]{geometry}
\usepackage[colorlinks, linkcolor=black, citecolor=blue, urlcolor=cyan]{hyperref}
\usepackage{natbib} 
\setcitestyle{authoryear, open={}, close={}}

\usepackage{booktabs} 
\usepackage{caption} 
\usepackage{tabularx} 
\usepackage[section]{placeins}

\usepackage{setspace} 
\usepackage{lipsum} 

\usepackage{amsfonts} 
\usepackage{bm} 

\usepackage{amsthm} 
\theoremstyle{plain}
\newtheorem{theorem}{Theorem}[section]
\newtheorem*{theorem*}{Theorem} 

\newtheorem{lemma}[theorem]{Lemma}
\newtheorem{proposition}[theorem]{Proposition}

\theoremstyle{definition}
\newtheorem{assumption}{Assumption}
\newtheorem{remark}{Remark}
\newtheorem{example}{Example}

\numberwithin{equation}{section}

\title{\MakeUppercase{Network double autoregression}}
\author{Tingting Li, Hao Wang}
\date{November 1, 2024}

\DeclareMathOperator{\BIC}{BIC}
\DeclareMathOperator{\Var}{Var}
\newcommand{\argmin}[1]{\mathop{\operatorname{argmin}}_{#1}}
\newcommand{\argmax}[1]{\mathop{\operatorname{argmax}}_{#1}}
\newcommand{\Diag}[1]{\operatorname{Diag}\left\{#1\right\}}

\newcommand{\asequal}{\stackrel{a.s.}{=\mathrel{\mkern -3mu}=}}
\newcommand{\x}{\times}
\newcommand{\R}{\mathbb{R}}

\newcommand{\bx}{\bm{x}}
\newcommand{\by}{\bm{y}}
\newcommand{\bz}{\bm{z}}
\newcommand{\bw}{\mathbf{w}}
\newcommand{\bone}{\bm{1}}

\begin{document}

\maketitle

\begin{abstract}
    Modeling high-dimensional time series with simple structures 
    is a challenging problem. 
    This paper proposes a network double autoregression (NDAR) model, 
    which combines the advantages of network structure 
    and the double autoregression (DAR) model, 
    to handle high-dimensional, conditionally heteroscedastic, 
    and network-structured data within a simple framework. 
    The parameters of the model are estimated 
    using quasi-maximum likelihood estimation, 
    and the asymptotic properties of the estimators are derived. 
    The selection of the model's lag order will be based on the Bayesian information criterion. 
    Finite-sample simulations show that the proposed model performs well 
    even with moderate time dimensions and network sizes. 
    Finally, 
    the model is applied to analyze three different categories of stock data. 
    \\[10pt]
\end{abstract}

\section{Introduction}

Time series models play a crucial role in fields such as finance, economics, and econometrics. 
Researchers often need to model time series data to reveal the correlation structure, 
as well as to forecast future trends in daily returns of 
stock prices, exchange rates, and other financial assets. 
Classical time series models include autoregressive (AR) models, 
moving average (MA) models, 
and their combined form, 
the autoregressive moving average models (ARMA). 
These models are primarily used to describe the lag effects 
and short-term random shocks within a time series. 
However, 
traditional models exhibit significant limitations in handling complex dependencies 
and conditional heteroscedasticity. 
As a result, 
the autoregressive conditional heteroskedasticity (ARCH) model was introduced by \cite{engle1982autoregressive}, 
and \cite{bollerslev1986generalized} extended it to the generalized ARCH (GARCH) model 
to effectively capture the volatility clustering characteristics in time series.

With the ongoing development of time series analysis, 
researchers have gradually expanded univariate models to multivariate settings, 
proposing models like vector autoregression (VAR) 
and multivariate GARCH (MGARCH) (\cite{bollerslev1990modelling, engle1995multivariate}). 
These models can capture the interdependencies among multi-dimensional time series. 
However, 
the complexity of parameter estimation and inference significantly increases in high-dimensional time series, 
particularly as the number of parameters grows polynomially with dimensions. 
To address these challenges, 
some researchers have introduced structural constraints, 
such as the constant or dynamic conditional correlation GARCH model (\cite{bollerslev1990modelling, tse2002multivariate}), 
or dimension-reduction techniques (\cite{engle1990asset, hu2014principal, li2016modeling}) 
to reduce parameter complexity 
and improve computational efficiency. 

To balance the simplicity and flexibility of time series models, 
the double autoregressive (DAR) model (\cite{ling2004estimation, ling2007double}), 
as a special case of the AR-ARCH models (\cite{weiss1986asymptotic}), 
has increasingly attracted researchers' attention. 
The DAR model, 
which simultaneously models the conditional mean and conditional variance of time series, 
provides a simple and powerful framework: 
\[
y_t = \sum_{i=1}^{p} \alpha_i y_{t-i} + \eta_t \sqrt{\omega + \sum_{i=1}^{q} \beta_i y_{t-i}^2}, 
\]
where the conditional expectation and variance dynamically depend only on past observations. 
It has been shown that the DAR model only requires the existence of fractional moments 
when using quasi-maximum likelihood estimation (QMLE). 
Lower moments of \(y_t\) imply a larger parameter space (\cite{ling2007double}). 
In contrast, 
the ARMA-GARCH model demands stricter asymptotic distribution requirements for QMLE, 
necessitating the existence of fourth-order moments (\cite{francq2004maximum}). 
Based on the novelty of the DAR model, 
various extensions have been proposed, 
such as other loss functions, 
quasi-maximum exponential likelihood estimators (\cite{zhu2013quasi}), 
and quantile loss (\cite{zhu2022quantile}) 
to weaken the model's assumptions. 
More varieties of DAR models also have been developed, 
including double AR model without intercept (DARWIN) (\cite{li2019double}), 
threshold double autoregressive (TDAR) models (\cite{li2015asymptotic, li2016threshold}), 
linear double autoregressive (LDAR) models (\cite{zhu2018linear, li2019double, tan2022asymmetric}), 
mixture double autoregressive models (\cite{li2017mixture}), 
and augmented double autoregressive models (\cite{jiang2020non}). 
Like MGARCH models, in further extending the DAR model to a multivariate setting, 
the vector DAR model (\cite{zhu2017vector}) 
and vector linear DAR model (\cite{lin2024vector}) have proposed.
But for higher dimensional sequences, 
how to deal with the curse of dimensionality remains a challenge.

In recent years, 
the introduction of network structure data 
has provided a new perspective for large-scale time series modeling. 
The network vector autoregression (NVAR), 
proposed by \cite{zhu2017network}, 
analyzes social network data by specifying attention-following relationships. 
Subsequently, 
various models have been proposed (\cite{zhu2019network, chen2023network, xu2024dynamic, armillotta2024count}). 
To broaden the application scope of single-network structures, 
models with grouped structures have also been introduced (\cite{zhu2020grouped, zhu2023simultaneous}). 
At the same time, 
GARCH-type models can also be extended (\cite{zhou2020network, pan2024threshold}). 
These models simplify and effectively reduce computational complexity 
from \(O(N^2)\) to \(O(N)\) 
for data involving \(N\) nodes 
by using the connectivity information among nodes in the network, 
effectively capturing specific correlation structures within high-dimensional time series data, 
when dealing with large node sequences.

This paper primarily investigates embedding network structures into the DAR model,
where both the mean and volatility components are embedded with the same network structure, 
enriching the model's application in higher dimensional time series. 
The remainder of this paper is organized as follows. 
Section 2 introduces first-order and higher orders NDAR model, stationarity and ergodicity conditions, 
Quasi-Maximum Likelihood Estimation, asymptotic results, and BIC for model selection; 
Section 3 demonstrates the model's estimation performance on small samples through numerical simulation experiments; 
Section 4 applies the model to real stock market data. 
Technical details are provided in the appendix.

\section{Network double autoregression}

\subsection{Model and notation}
Consider an adjacency matrix \( A = (a_{ij}) \in \mathbb{R}^{N \times N} \) 
representing a network with \( N \) nodes, 
where \( a_{ij} = 1 \) if node \( i \) is connected to node \( j \), 
and \( a_{ij} = 0 \) otherwise. 
The matrix is defined such that \( a_{ii} = 0 \) for all \( 1 \leq i \leq N \). 
In practical applications, 
\( a_{ij} = 1 \) indicates that node \( i \) is influenced by node \( j \), 
and mutual influence between nodes \( i \) and \( j \) is represented by \( a_{ij} = a_{ji} = 1 \).
Each node \(  i  \) at time point \( t \) is associated with a response \( y_{it}\in\R^{1} \).
Then the network double autoregression model with order \( (1, 1) \), or called NDAR\( (1,1) \), is defined as
\begin{equation}\label{eq:NDAR_it}
  \begin{aligned}
    y_{it} = \alpha n_{i}^{-1} \sum_{j=1}^{N} a_{ij} y_{j,t-1} + \beta y_{i, t-1} 
            + \eta_{it}\sqrt{\omega + \phi n_{i}^{-1} \sum_{j=1}^{N} a_{ij} y_{j,t-1}^2 + \psi y_{i,t-1}^{2}}
  \end{aligned}\ ,
\end{equation}
 where \( n_{i} = \sum_{j=1}^{N}a_{ij} \), particularly delete node \( i \) if \( n_{i}=0 \),
 and the collection \(\{\eta_{it}\}\) consists of independent and identically distributed random variables 
 with \( E\eta_{it} = 0 \) and \( \Var(\eta_{it}) = 1 \) for parameter identifiability.
The \( y_{it} \) is influenced by four parts. 
In the mean component, 
 the term \( n_{i}^{-1} \sum_{j=1}^{N} a_{ij} y_{j,t-1} \) represents the average impact from neighbors, 
 and \( y_{j,t-1} \) is the standard autoregressive term. 
In the volatility component, 
 the terms \( n_{i}^{-1} \sum_{j=1}^{N} a_{ij} y_{j,t-1}^2 \) and \( y_{i,t-1}^{2} \) 
 are also the average impact from neighbors and the standard autoregressive term, respectively. 
The parameters \( \alpha, \beta, \omega, \phi, \) and \( \psi \) are unknown.
To ensure positivity in the square root,
  \( \omega > 0 \) and \( \phi, \psi \geq 0 \) are assumed.

For convenience, define the following notations: 
 \( \bm{y}_{t} = (y_{1t}, \ldots, y_{Nt})' \), 
 \( x_{it} = y_{it}^2 \), 
 \( \bm{x}_{t} = (x_{1t}, \ldots, x_{Nt})' \), 
 \( \bm{\eta}_{t} = (\eta_{1t}, \ldots, \eta_{Nt})' \), 
 and the row-normalized adjacency matrix \( W = (\bw_1, \ldots, \bw_N)' \), 
 where \( \bm{w}_{i} = n_{i}^{-1}(a_{i1}, \ldots, a_{iN})' \). 
Then the model \eqref{eq:NDAR_it} can be expressed as
\begin{equation}\label{eq:NDAR_it_vec}
  y_{it} = \alpha \bw_{i}'\by_{t-1} + \beta y_{i, t-1}
     + \eta_{it}\sqrt{\omega + \phi \bw_{i}'\bx_{t-1} + \psi x_{i,t-1}},
\end{equation}
 or a more compact form
\begin{equation}\label{eq:NDAR_it_vec_compact}
  \begin{cases}
    \begin{aligned}
      \by_{t} &= \alpha W\by_{t-1} + \beta \by_{t-1} + H_{t}^{1/2}\bm{\eta}_{t}, \\
      H_t &= \Diag{\omega\bone + \phi W\bm{x}_{t-1} + \psi \bm{x}_{t-1}},
    \end{aligned}
  \end{cases}
\end{equation}
 where
 \( \bone = \bone_N \) is an \( N \times 1 \) vector of ones 
 and \( \Diag{\bm{a}} \) denotes a diagonal matrix with \( \bm{a} \) as its diagonal.
To model more general cases, 
 the NDAR\( (p, q) \) model is defined as
 \begin{equation}\label{eq:NDAR_pq}
  \begin{cases}
    \begin{aligned}
      \by_{t} &= \sum_{r=1}^{p} \alpha_r W\by_{t-1} + \sum_{r=1}^{q}\beta_r \by_{t-1} + H_{t}^{1/2}\bm{\eta}_{t}, \\
      H_t &= \Diag{\omega\bone + \sum_{r=1}^{p} \phi_r W\bm{x}_{t-1} + \sum_{r=1}^{q}\psi_r \bm{x}_{t-1}}.
    \end{aligned}
  \end{cases}
 \end{equation}

  Although the NDAR\( (p, q) \) model defined as \eqref{eq:NDAR_pq} 
  use same order \( p \) (\( q \)) in both the mean and volatility components,
  it is possible to consider different orders.
  In applications, one can use \( p_1 \) (\( q_1 \)) in the mean component 
  and \( p_2 \) (\( q_2 \)) in the volatility component
  but condition \( p_1\leq p_2 \) (\( q_1\leq q_2 \))  must be satisfied, 
  as can be seen in proof section, 
  to ensure that volatility term can control the mean term of
  the log-likelihood function, score function, and information matrix.

\subsection{Strict stationarity and ergodicity}
It can be observed that model \ref{eq:NDAR_pq} is a special case of the following model:
\begin{equation}
  \bm{y}_t = f(\bm{y}_{t-1}, \ldots, \bm{y}_{t-p}) + H(\bm{y}_{t-1}, \ldots, \bm{y}_{t-p})\bm{\eta}_t,
\end{equation}
where \( f \) is a vector-valued function, and \( H \) is a matrix-valued function.
Based on the conditions (B1)-(B4) and Theorem 1 in \cite{lu2001l1}
the following proposition can be obtained.

\begin{proposition}\label{prop:strictly_stationary}
  The NDAR\( (p, q) \) model defined as \eqref{eq:NDAR_pq} is strictly stationary and geometrically ergodic
  if 

  (i) the random vectors \( \{\bm{\eta}_t\} \) are independent identically distributed, 
  absolutely continuous and have a density which is
  positive almost everywhere with \( E(\Vert\bm{\eta}_t\Vert_{1})<\infty \);

  (ii) \( \max_{j=1,\ldots,N}\left\{ \varUpsilon_{j} \right\}<1 \), where
  \[ 
    \varUpsilon_{j} = \sum_{r=1}^{p}|\alpha_{r}|\left( \sum_{i=1}^{N}\frac{a_{ij}}{n_{i}} \right) + \sum_{r=1}^{p}|\beta_{r}| + E|\eta|\left[ \sum_{r=1}^{p}\sqrt{\phi_{r}}\left( \sum_{i=1}^{N}\frac{a_{ij}}{\sqrt{n_{i}}} \right) + \sum_{r=1}^{p}\sqrt{\psi_{r}}  \right].
  \]
\end{proposition}

\section{Model estimation}
\subsection{Quasi-maximum likelihood estimation}
In this subsection, we let \( (p,q) \) be the true order for convinence,
the choice of order is considered in the next subsection.
Denote model parameters \( \theta = (\mu', \sigma')\in\R^{2p+2q+1} \),
where \( \mu = (\alpha_1, \ldots, \alpha_p, \beta_1, \ldots, \beta_q)' \)
and \( \sigma = (\omega, \phi_1, \ldots, \phi_p, \psi_1, \ldots, \psi_q)' \).
The true parameter is denoted by \( \theta_0 = (\mu_0', \sigma_0') \),
where \( \mu_0 = (\alpha_{10}, \ldots, \alpha_{p0}, \beta_{10}, \ldots, \beta_{q0})' \)
and \( \sigma_0 = (\omega_0, \phi_{10}, \ldots, \phi_{p0}, \psi_{10}, \ldots, \psi_{q0})' \).
Let \( m=\max\{p, q\} \),
assume that response \( \{\by_{1-m},\ldots, \by_{1}, \ldots, \by_{T}\} \)
come from a NDAR\( (p, q) \) process
with true parameter \( \theta_0 \) and a fixed adjacency matrix \( A \).
Denote the regressive variables about node \( i \) at time \( t \) by
\begin{equation*}
    \begin{aligned}
        \bz_{1it} & = \left( \bw'_{i}\by_{t-1}, \ldots, \bw'_{i}\by_{t-p}, y_{i,t-1}, \ldots, y_{i,t-q} \right)' \text{ and} \\
        \bz_{2it} & = \left( 1, \bw'_{i}\bx_{t-1}, \ldots, \bw'_{i}\bx_{t-p}, x_{i,t-1}, \ldots, x_{i,t-q} \right)',         \\
    \end{aligned}
\end{equation*}
then the quasi-maximum likelihood estimator (QMLE) is defined as
\begin{equation}\label{eq:qmle}
    \hat{\theta} = \argmax{\theta\in\Theta}\sum_{t=1}^{T}L_t(\theta),\;
    L_t(\theta) = -\frac{1}{2}\sum_{i=1}^{N}\left( \log{h_{it}(\sigma)}
    + \frac{\epsilon_{it}^{2}(\mu)}{h_{it}(\sigma)} \right),
\end{equation}
where \( \Theta \) is the parameter space,
\( \varepsilon_{it}(\mu) = y_{it} - \bz_{1it}'\mu \) and
\( h_{it}(\sigma) = \bz_{2it}'\sigma \).

\begin{assumption}\label{assumption:y}
    \( \{\bm{y}_t\} \) is strictly stationary and ergodic,
    and exists \( u>0 \) such that \( E\left\lVert \by_t \right\rVert^u<\infty \).
\end{assumption}
\begin{assumption}\label{assumption:space}
    The parameter space \( \Theta \) is compact
    and \( \theta_0=(\mu_0', \sigma_0')' \) is an interior point of \( \Theta \)
    with \( |\alpha_{i}|, |\beta_{j}| < \bar{b}_1 \)
    and \( 0 < \underline{b}_2 < \omega, \phi_{i}, \psi_{j} < \bar{b}_2 \)
    for all \( i=1,\ldots,p \) and \( j=1,\ldots,q \).
\end{assumption}
\begin{assumption}\label{assumption:epsilon}
    \( \{\eta_{it}\} \) is i.i.d. across \( i \) and \( t \) with zero mean and unit variance.
    and the matrix
    \( D = \begin{pmatrix}
        1 & \frac{\kappa_3}{\sqrt{2}} \\ \frac{\kappa_3}{\sqrt{2}} & \frac{\kappa_4-1}{2}
    \end{pmatrix} \) is positive definite,
    where \( \kappa_3 = E\eta_{it}^{3} \) and \( \kappa_4 = E\eta_{it}^{4}<\infty \).
\end{assumption}

\begin{theorem}\label{theorem:normal}
    If Assumptions \ref{assumption:y}--\ref{assumption:epsilon} hold,
    then \( \sqrt{NT}(\hat{\theta} - \theta_0) \rightarrow_{d} \mathcal{N}(\Omega^{-1}_{0}\Sigma_{0}\Omega^{-1}_{0}) \),
    where \( \rightarrow_{d} \) denote convergence in distribution,
    \begin{equation*}
            \Omega_0 = \frac{1}{N}\sum_{i=1}^{N}E \left[ \varGamma_{it}(\theta_0)\varGamma'_{it}(\theta_0) \right], \quad 
            \Sigma_0 = \frac{1}{N}\sum_{i=1}^{N}E \left[ \varGamma_{it}(\theta_0) D \varGamma'_{it}(\theta_0) \right], 
    \end{equation*}
    and 
    \[
        \varGamma_{it}(\theta_0) = \Diag{\frac{\bz_{1it}}{\sqrt{h_{it}(\sigma_0)}}, \frac{\bz_{2it}}{\sqrt{2}h_{it}(\sigma_0)}}
    .\]
\end{theorem}

When one apply this approximated distribution to make statistical inference, 
\( \Omega_0, \Sigma_0 \) should be replaced by 
\[
    \hat{\Omega} = \frac{1}{N}\sum_{i=1}^{N}E \left[ \varGamma_{it}(\hat{\theta}) \varGamma'_{it}(\hat{\theta})\right], \quad
    \hat{\Sigma} = \frac{1}{N}\sum_{i=1}^{N}E \left[ \varGamma_{it}(\hat{\theta}) \hat{D} \varGamma'_{it}(\hat{\theta}) \right],
\]
where 
\begin{equation*}
    \begin{aligned}
        &\varGamma_{it}(\hat\theta) = \Diag{\frac{\bz_{1it}}{\sqrt{h_{it}(\hat\sigma)}}, \frac{\bz_{2it}}{\sqrt{2}h_{it}(\hat\sigma)}} , \\
        &\hat{D} = \begin{pmatrix}
            1 & \displaystyle\frac{1}{\sqrt{2} N T} \sum_{i=1}^{N} \sum_{t=1}^{T} \frac{\epsilon_{it}^{3}(\hat{\mu})}{h_{it}^{3/2}(\hat{\sigma})} \\
            \displaystyle \frac{1}{\sqrt{2} N T} \sum_{i=1}^{N} \sum_{t=1}^{T} \frac{\epsilon_{it}^{3}(\hat{\mu})}{h_{it}^{3/2}(\hat{\sigma})} & \displaystyle\frac{1}{2 N T} \sum_{i=1}^{N} \sum_{t=1}^{T} \frac{\epsilon_{it}^{4}(\hat{\mu})}{h_{it}^{2}(\hat{\sigma})} - \frac{1}{2}
        \end{pmatrix}.
    \end{aligned}
\end{equation*}


\subsection{Model selection}

In this subsection, we discuss how to select the optimal order \( (p, q) \) for the NDAR model. 
Some notations need to be defined.
Let \( (p_0, q_0) \) represent the true order 
and the corresponding true parameters as 
\( \theta_0^{p_0, q_0} \).
When estimating the parameters using an order \( (p, q) \), 
the resulting estimator is denoted by \( \hat{\theta}^{p, q} \). 
When we overestimate the order \( (p, q) \), 
we are actually estimating \( \theta \) in a higher-dimentional parameter space. 
Let \( \iota_{p, q}(\theta^{p_0, q_0}) \) be defined as a \( 2p+2q+1 \) dimentional vector 
by placing zeros in extra positions of \( \theta^{p_0, q_0} \). 
The notation \( \pi_{p_0, q_0}(\theta^{p, q}) \) denotes the projection of 
\( \theta^{p, q} \) onto its non-extra positions. 
Formally, we can write:
\begin{equation*}
    \begin{aligned}
        \iota_{p, q}(\theta^{p_0, q_0}) & =
        (\alpha'_{(p_0)}, 0'_{(p - p_0)},
        \beta'_{(q_0)}, 0'_{(q - q_0)}, \omega,
        \phi'_{(p_0)}, 0'_{(p - p_0)},
        \psi'_{(q_0)}, 0'_{(q - q_0)})', \\
        \pi_{p_0, q_0}(\theta^{p, q})   & =
        (\alpha_{(p_0)}, \beta_{(q_0)}, \omega,
        \phi_{(p_0)}, \psi_{(q_0)})',
    \end{aligned}
\end{equation*}
where \( 0_{(k)} \) is a zero vector of length \( k \),
\( \alpha_{(k)} \) represents \( (\alpha_1, \ldots, \alpha_k)' \), 
and similarly, \( \beta_{(k)}, \phi_{(k)}, \) and \( \psi_{(k)} \) 
are defined in the same way.

To determine the appropriate order \( (p, q) \), 
we utilize the Bayesian Information Criterion (BIC), defined as:
\[
    \BIC(p, q) = -2F_T(\hat{\theta}^{(p, q)}) + (2p + 2q + 1) \ln T,
\]
where \( F_T(\theta) = \sum_{t=1}^{T} L_t(\theta) \). 
Let \( r_m \) be a fixed positive integer that do not lower than \( \max\{p_0,q_0\} \), 
and define \( \mathcal{C} = \{(p, q) \mid 0 \leq p, q \leq r_m\} \) 
as the possible combinations of \( p \) and \( q \). 
The optimal order \( (\hat{p}, \hat{q}) \) is then selected as:
\begin{equation}\label{eq:bic}
    (\hat{p}, \hat{q}) = \argmin{(p, q) \in \mathcal{C}} \BIC(p, q).
\end{equation}

To apply the Bayesian Information Criterion (BIC), 
we need to show that additional parameters do not bring 
higher likelihood values, in the \( O_p \) sense, 
when both \( p \geq p_0 \) and \( q \geq q_0 \). 
In this cases,
the estimator \( \hat{\theta}^{p, q} \) is a maximum 
found in a higher-dimensional parameter space instead of the true space,
where \( \iota_{p, q}(\theta^{p_0, q_0}) \) is the true parameter. 
Therefore, due to some elements of \( \hat{\theta}^{p, q} \) 
lying on the boundary, 
\( \hat{\theta}^{p, q}-\iota_{p, q}(\theta_0^{p_0, q_0}) \) should not be approximated 
by a normal distribution.

Let \( \Theta^{p,q} \) be the subset of \( \R^{2p+2q+1} \)
and \( \pi_{p_0,q_0}(\Theta^{p,q})=\{\pi_{p_0,q_0}(\theta^{p,q})|\theta^{p,q}\in\Theta^{p,q}\} \)
\begin{assumption}\label{assumption:space1}
    \( \Theta^{r_m,r_m} \) is compact
    with \( 0 < \underline{\sigma}<\omega< \bar{\sigma} \) 
    and \( 0 \leq \phi_{i}, \psi_{j} < \bar{\sigma} \)
    for all \( i=1,\ldots,p \) and \( j=1,\ldots,q \);
    \( \theta_{0}^{p_0,q_0} \) is an intieror of \( \pi_{p_0,q_0}(\Theta^{r_m,r_m}) \)
    with \( \underline{\sigma} < \phi_{i}, \psi_{j} < \bar{\sigma} \)
    for all \( i=1,\ldots,p_0 \) and \( j=1,\ldots,q_0 \),
\end{assumption}
Namely, in the highest-dimentional parameter space \( \Theta^{r_m,r_m} \),
volatility coeffients (except for \( \omega \)) are allowed equal zero.
In correct dimentional parameter space \( \pi_{p_0,q_0}(\Theta^{r_m,r_m}) \), however,
the condition is identical with assumption \ref{assumption:space}.

\begin{lemma}\label{lemma:higher_space_estimator}
    Assume that \( E\left\lVert \by_t \right\rVert^6<0 \),
    The conditions of theorem \ref{theorem:normal} hold 
    except for the assumption \ref{assumption:space} is replaced by assumption \ref{assumption:space1}, 
    \[
        \hat{\theta}^{p,q} = \argmax{\theta^{p,q}\in\pi_{p,q}(\Theta^{r_m,r_m})} \sum_{t=1}^{T} L_t(\theta^{p,q})
    \]
    then \( \sum_{t=1}^{T} L_t(\hat{\theta}^{p,q}) - \sum_{t=1}^{T} L_t(\iota_{p,q}(\theta_{0}^{p_0,q_0}))  = O_p(1)  \).
\end{lemma}
\begin{remark}
    In practice, the algorithm for searching \( \hat{\theta} \) in \eqref{eq:qmle} 
    must allow the lower bounds of \( \phi_i \) and \( \psi_j \) to be zero, 
    though Assumption \ref{assumption:space} needs a sufficiently low bound \( \underline{b}_2 \).
    In real data, the elements of estimator usually be a small number 
    but not exactly equal zero.
\end{remark}

\begin{theorem}\label{theorem:selector}
    Under the Assumptions of Lemma \ref{lemma:higher_space_estimator}, 
    \( P\{(\hat{p},\hat{q})=(p_0,q_0)\}\to 1 \) as \( T\to\infty \).
\end{theorem}

\section{Simulation studies}

\subsection{Model Estimation}
To demonstrate the performance of the proposed estimator with a finite sample size, 
we consider the following setup.  
we assume \( \eta_{it} \sim \mathcal{N}(0, 1) \) or \( \eta_{it} \sim \sqrt{3/5}\ t_{5} \), 
where \( t_{5} \) denotes the Student's \( t \)-distribution with 5 degrees of freedom, and the constant 
\( \sqrt{3/5} \) is chosen to ensure \( \mathbb{E}[\eta_{it}^2] = 1 \).
To ensure that the simulated data \( \{y_{it}\} \) approaches a near-stationary state, 
it is necessary to generate data over a longer time span \( T_0 + T \), 
retaining only \( \{y_{it}, t > T_0\} \). In this simulation, we set \( T_0 = 500 \) 
to ensure stationarity, as a larger \( T_0 \) does not significantly change the results. 
Since the network matrix \( A \) in model \eqref{eq:NDAR_pq} is predefined, 
all simulated network matrices are generated once and then fixed.
There are three distinct scenarios, 
each demonstrating a unique structural property of the network.

\begin{example}\label{ex:1}
    (Network with Randomly Distributed Structure)
    In this example, each node's out-degree is determined by a discrete uniform
    distribution, and target nodes are selected using simple random sampling.
    Specifically, the out-degree \( u_i \) of node \( i \) is sampled 
    from the set \( \{1, 2, 3, 4, 5\} \) with equal probability. 
    The set of target nodes
    \( \mathcal{S}_i \) is constructed by randomly selecting \( u_i \) nodes from
    \( \{1, \dots, N\} \setminus \{i\} \) without replacement. 
    Then the network matrix \( A \) is defined such that \( a_{ij} = 1 \) if \( j \in \mathcal{S}_i \),
    and \( a_{ij} = 0 \) otherwise. 
    This network structure is illustrated in Figure \ref{fig:case1_hist_graph}.
\end{example}
\begin{figure}[htbp]
    \centering
    \begin{subfigure}{0.45\textwidth}
        \includegraphics[width=\textwidth]{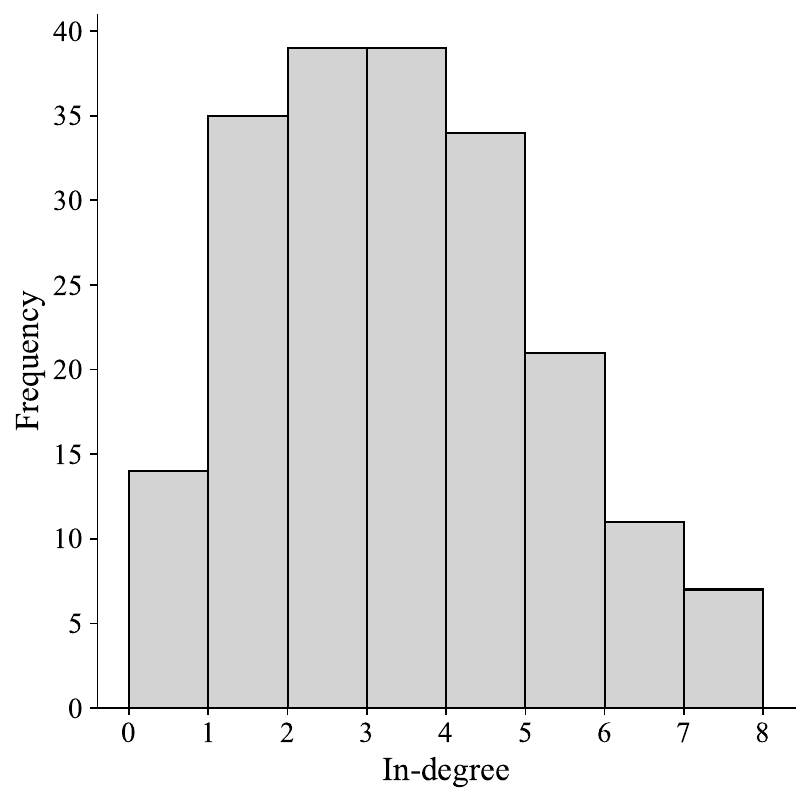}
        \caption{}
    \end{subfigure}
    \begin{subfigure}{0.45\textwidth}
        \includegraphics[width=\textwidth]{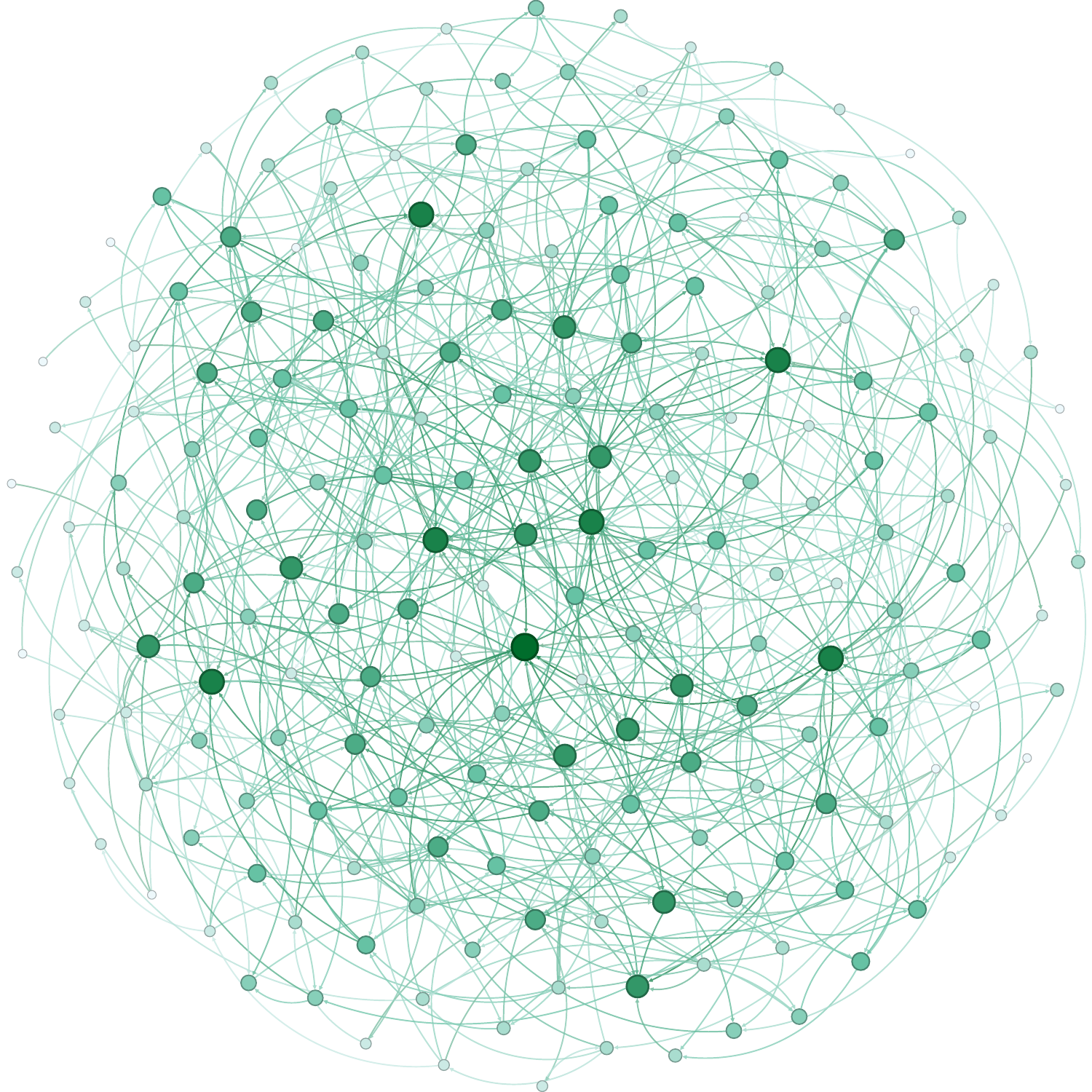}
        \caption{}
    \end{subfigure}
    \caption{Simulated Network for Example \ref{ex:1} with \( N = 200 \) nodes. 
    (a) In-degree distribution, illustrating that each node exhibits relatively uniform influence. 
    (b) Network graph, where deeper colors of nodes or edges represent higher in-degree values.}      
    \label{fig:case1_hist_graph}
\end{figure}

\begin{example}\label{ex:2}
    (Network with Power-Law Degree Distribution)
    Similar to the first example, the out-degree \( u_i \) of each node \( i \)
    is sampled from the set \( \{1, 2, 3, 4, 5\} \) with equal probability. 
    To simulate the presence of a few highly influential nodes, however, 
    the selection of target nodes are binded by a power-law distribution. 
    Specifically, for each node \( i \), we generate a random variables,
    \( v_i \sim P(x) = c x^{-\gamma}, x \in \mathbb{Z}_+ \), where \( c \) is a
    normalization constant, \( \mathbb{Z}_+ \) represents the set of positive
    integers and \( \gamma = 2.5 \) is seted. 
    The target node set \( \mathcal{S}_i \) is constructed by randomly
    selecting \( u_i \) nodes from \( \{1, \dots, N\} \setminus \{i\} \) with
    probabilities proportional to \( v_i \). The network matrix is defined such
    that \( a_{ij} = 1 \) if \( j \in \mathcal{S}_i \), and \( a_{ij} = 0 \)
    otherwise. This network structure is depicted in Figure \ref{fig:case2_hist_graph}.
\end{example}
\begin{figure}[htbp]
    \centering
    \begin{subfigure}{0.45\textwidth}
        \includegraphics[width=\textwidth]{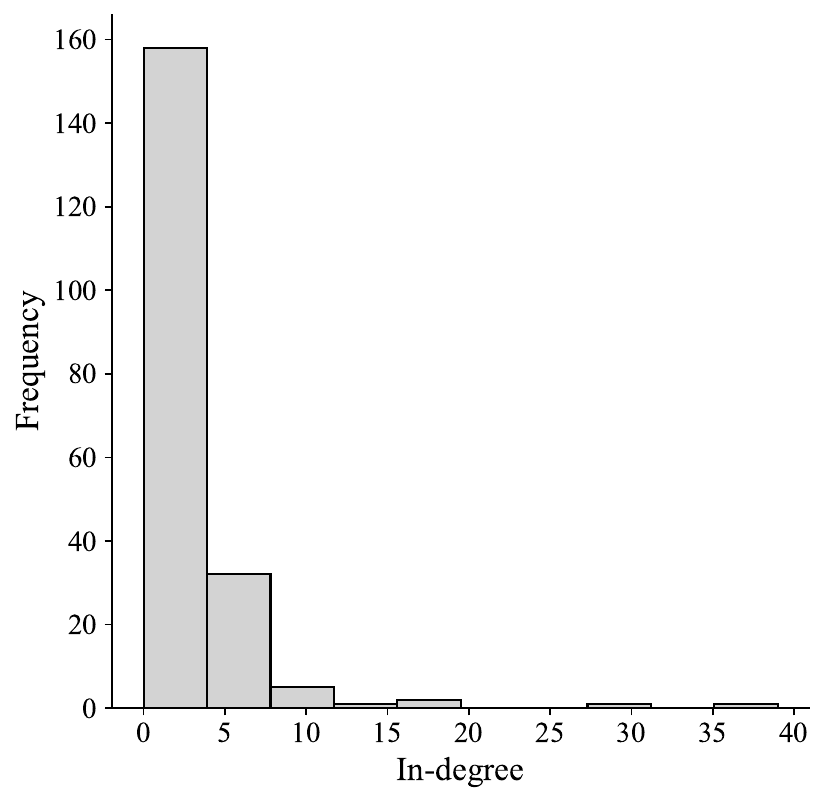}
        \caption{}
    \end{subfigure}
    \begin{subfigure}{0.45\textwidth}
        \includegraphics[width=\textwidth]{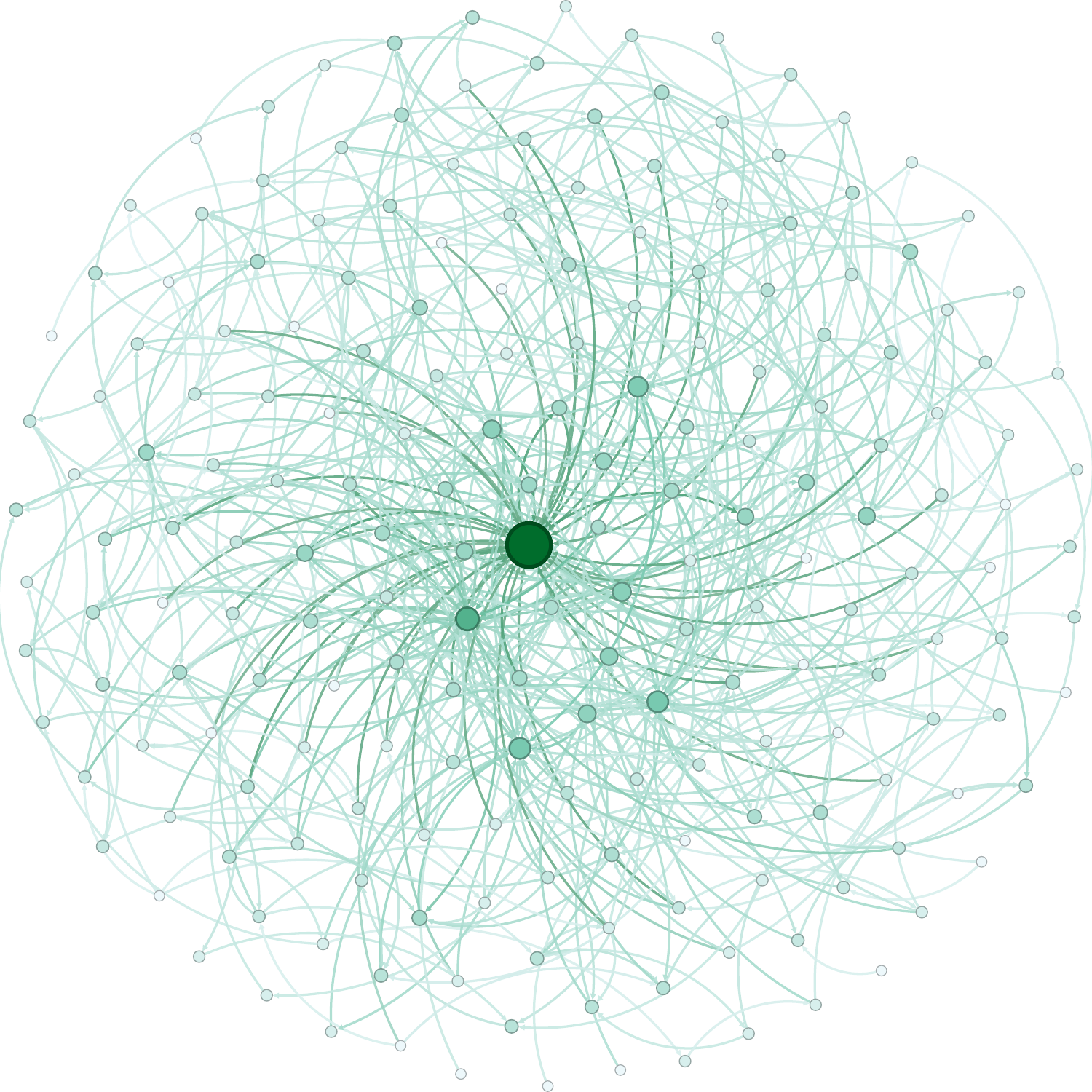}
        \caption{}
    \end{subfigure}
    \caption{Simulated Network for Example \ref{ex:2} with \( N = 200 \) nodes. 
    (a) In-degree distribution, showing that most nodes have minimal influence, while a few nodes influence more than 20 other nodes. 
    (b) Network graph, where deeper colors of nodes or edges represent higher in-degree values, emphasizing the presence of super nodes.}    
    \label{fig:case2_hist_graph}
\end{figure}

\begin{example}\label{ex:3}
    (Stochastic Block Network Structure)
    In this example, nodes are divided into \( K = N/10 \) blocks, and connections
    are more likely to occur between nodes within the same block than between nodes
    in different blocks. Specifically, the probability of an edge existing between
    nodes \( i \) and \( j \) is given by \( P\{a_{ij} = 1\} = p_1 \) if nodes \( i \)
    and \( j \) belong to the same block, and \( P\{a_{ij} = 1\} = p_2 \) if they
    belong to different blocks, where \( p_1 > p_2 \). To ensure consistency in
    network density compared to other structures, the parameters are set as
    \( p_1 = 28/N \) and \( p_2 = 0.01/N \). This network structure is visualized
    in Figure \ref{fig:case3_hist_graph}.
\end{example}
\begin{figure}[htbp]
    \centering
    \begin{subfigure}{0.45\textwidth}
        \includegraphics[width=\textwidth]{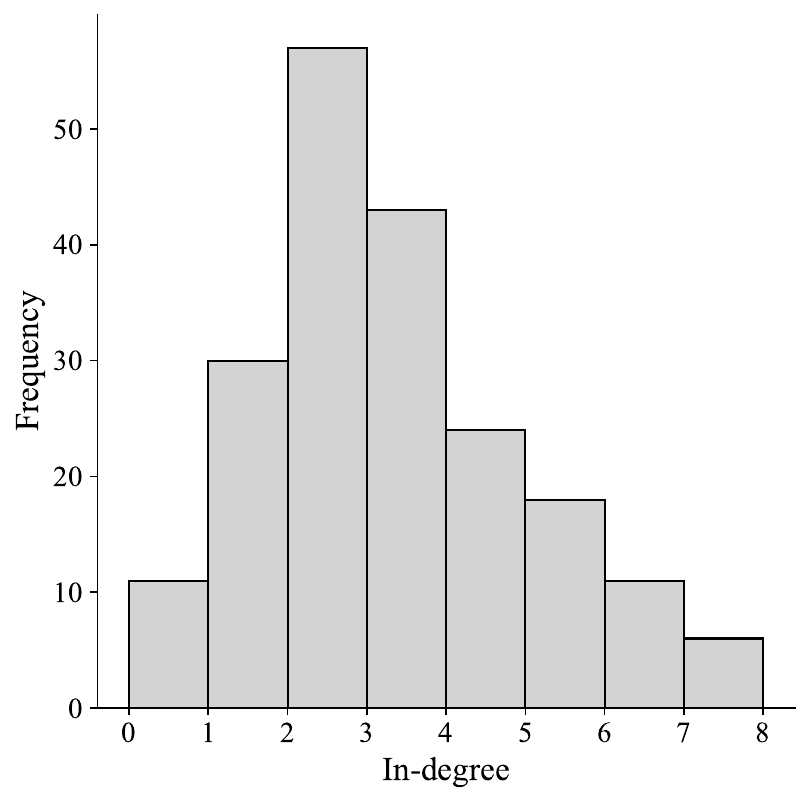}
        \caption{}
    \end{subfigure}
    \begin{subfigure}{0.45\textwidth}
        \includegraphics[width=\textwidth]{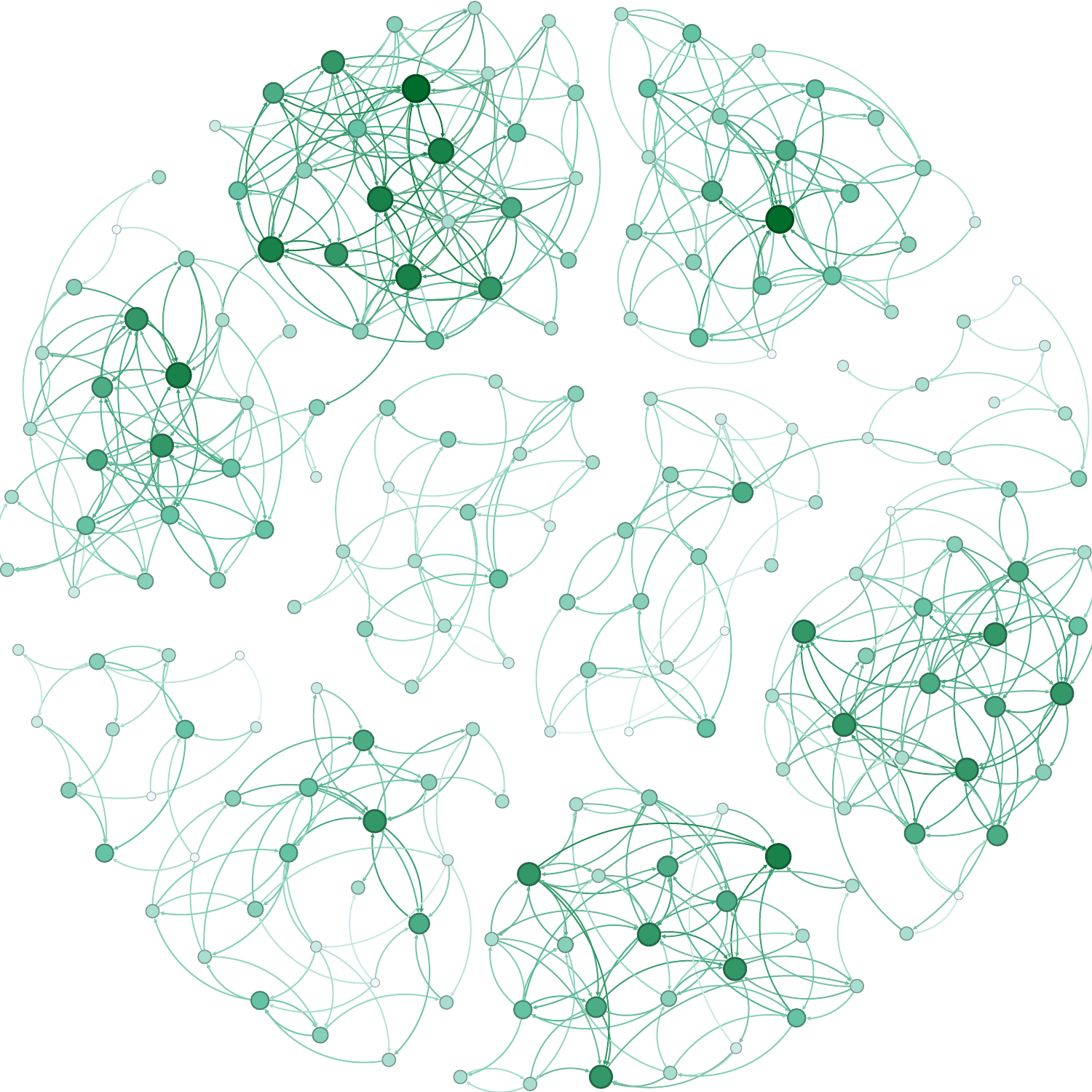}
        \caption{}
    \end{subfigure}
    \caption{Simulated Network for Example \ref{ex:3} with \( N = 200 \) nodes. 
    (a) In-degree distribution, illustrating that most nodes influence 2-4 other nodes. 
    (b) Network graph, revealing approximately 10 community structures, 
    where deeper colors of nodes or edges represent higher in-degree values.}
    \label{fig:case3_hist_graph}
\end{figure}

We generate \( 1000 \) replications of nodes size \( N \in \{50, 100, 150\} \) and time span \( T \in \{100, 200, 400\} \).
The data are generated from the model \eqref{eq:NDAR_pq} with the following data generating procedures (DGPs):
\begin{align*}
    &\text{DGP 1:} & \quad (p_0, q_0) = (1, 1)\text{ with }\theta_0 &= (\alpha_{10}, \beta_{10}, \omega_0, \phi_{10}, \psi_{10})' \\
    &              &                                               &= (0.05, -0.1, 0.05, 0.05, 0.1)', \\
    &\text{DGP 2:} & \quad (p_0, q_0) = (1, 2)\text{ with }\theta_0 &= (\alpha_{10}, \beta_{10}, \beta_{20}, \omega_0, \phi_{10}, \psi_{10}, \psi_{20})' \\
    &              &                                               &= (0.05, -0.05, 0.1, 0.1, 0.05, 0.1, 0.1)',
\end{align*}
where DGP 1 is use for showing the performance of the QMLE under the true order,
while DGP 2 is used to evaluate the BIC in selecting the true order.
Republications \( R = 1000 \) are set for each scenario.

\subsection{Performance measurements}
In this subsection, we introduce the performance measurements used to evaluate the QMLE and BIC, respectively.
Then we present the simulation results for the two DGPs in Example \ref{ex:1}, \ref{ex:2} and \ref{ex:3}.

To evaluate the performance of the QMLE in DGP 1, several measurements are used to assess 
estimation accuracy. Let \( \hat{\theta}^{(r)} \) denote the \( r \)-th replication estimator 
of the true parameter vector \( \theta_{0} \), 
and let the subscript \( k \) of a vector \( \theta \) denote 
its \( k \)-th element.
Then the bias of the estimators is defined as
\( \text{Bias}_{k} = R^{-1} \sum_{r=1}^{R} (\hat{\theta}_{k}^{(r)} - \theta_{k0}) \).
The average of asymptotic standard deviation is defined as 
\( \text{ASD}_{k} = R^{-1} \sum_{r=1}^{R} \widehat{\text{SD}}_k \), 
where \( \widehat{\text{SD}}_k \) is the \( k \)-th diagonal element of
estimated asymptotic covariance matrix \( \hat\Omega^{-1}\hat\Sigma\hat\Omega^{-1} \).
The empirical standard deviation is defined as 
\( \text{ESD}_{k} = \{R^{-1} \sum_{r=1}^{R} (\hat{\theta}_k^{(r)} - \bar{\hat{\theta}}_k)^2\}^{1/2} \), 
where \( \bar{\hat{\theta}}_k = R^{-1} \sum_{r=1}^{R} \hat{\theta}_k^{(r)} \). 
The coverage probability of the 95\% empirical coverage probability is defined as 
\( \text{CP}_k = R^{-1} \sum_{r=1}^{R} I\{ \theta_{k0} \in \text{CI}_k^{(r)} \} \), 
where \( \text{CI}_k^{(r)} = [\hat{\theta}_k^{(r)} - z_{0.975} \widehat{\text{SD}}_k, \hat{\theta}_k^{(r)} + z_{0.975} \widehat{\text{SD}}_k] \),
\( z_{0.975} \) is the 97.5\% quantile of the standard normal distribution and
\( I\{\cdot\} \) is the indicator function.
Table \ref{tab:sim_bias_DGP1_case1} shows that as \( T \) increases, 
the Bias approaches 0, and both ASD and ESD decrease and becoming closer to each other. 
The CP consistently remains close to the theoretical value of 0.95. 
For the results corresponding to the heavy-tailed \( t_t \) distribution, both the Bias and ASD are larger compared to those of the standard normal distribution.

\begin{table}[h]
    \centering
    \caption{Simulation results for DGP 1 in Example \ref{ex:1}.
        The table reports the Biases(\(\x10^3\)), ASDs(\(\x10^2\)), ESDs(\(\x10^2\)),
        and CPs of the QMLE, where \( \eta_{it} \) follow the standard normal distribution
        or the scaled \( t_5 \) distribution.
        The network densities corresponding to each \( N \) are provided in parentheses below.}
    \begin{tabular*}{\textwidth}{@{\extracolsep{\fill}}cccrcccrccc}
        \toprule
        &     &              & \multicolumn{4}{c}{Normal} & \multicolumn{4}{c}{\( t_5 \)}                                           \\
        \cmidrule(lr){4-7} \cmidrule(lr){8-11}
        \( N \)      & \( T \)   &    QMLE      & Bias                       & ASD                           & ESD  & CP   & Bias & ASD  & ESD  & CP   \\
        \midrule
        50 & 100 & \(\alpha_1\) & 0.03 & 2.02 & 2.05 & 0.95 & -1.48 & 2.08 & 2.03 & 0.96 \\
        (5.80\%) &  & \(\beta_1\) & -0.21 & 1.52 & 1.50 & 0.95 & 0.77 & 1.60 & 1.63 & 0.94 \\
        &  & \(\omega\) & 0.07 & 0.18 & 0.18 & 0.95 & 0.09 & 0.27 & 0.29 & 0.94 \\
        &  & \(\phi_1\) & -0.51 & 2.12 & 2.12 & 0.94 & 1.40 & 2.97 & 4.14 & 0.92 \\
        &  & \(\psi_1\) & -0.78 & 1.81 & 1.81 & 0.94 & -1.13 & 3.02 & 3.37 & 0.94 \\
        & 200 & \(\alpha_1\) & 0.45 & 1.43 & 1.40 & 0.96 & 0.60 & 1.47 & 1.45 & 0.95 \\
        &  & \(\beta_1\) & 0.05 & 1.07 & 1.06 & 0.95 & 0.57 & 1.13 & 1.09 & 0.96 \\
        &  & \(\omega\) & 0.05 & 0.13 & 0.13 & 0.94 & -0.00 & 0.20 & 0.21 & 0.95 \\
        &  & \(\phi_1\) & -0.05 & 1.50 & 1.51 & 0.94 & -0.36 & 2.11 & 2.08 & 0.95 \\
        &  & \(\psi_1\) & -1.17 & 1.28 & 1.30 & 0.94 & 0.81 & 2.19 & 2.22 & 0.94 \\
        & 400 & \(\alpha_1\) & 0.17 & 1.01 & 1.04 & 0.93 & -0.37 & 1.04 & 1.02 & 0.95 \\
        &  & \(\beta_1\) & -0.33 & 0.76 & 0.75 & 0.95 & -0.57 & 0.80 & 0.79 & 0.95 \\
        &  & \(\omega\) & 0.02 & 0.09 & 0.09 & 0.94 & -0.08 & 0.14 & 0.14 & 0.94 \\
        &  & \(\phi_1\) & -0.26 & 1.06 & 1.06 & 0.94 & 0.10 & 1.52 & 1.58 & 0.96 \\
        &  & \(\psi_1\) & 0.13 & 0.91 & 0.91 & 0.94 & 0.44 & 1.55 & 1.56 & 0.95 \\
        200 & 100 & \(\alpha_1\) & 0.41 & 1.03 & 1.03 & 0.95 & 0.21 & 1.06 & 1.02 & 0.97 \\
        (1.48\%) &  & \(\beta_1\) & 0.40 & 0.76 & 0.77 & 0.95 & 0.06 & 0.80 & 0.84 & 0.94 \\
        &  & \(\omega\) & 0.05 & 0.09 & 0.09 & 0.95 & 0.06 & 0.14 & 0.15 & 0.96 \\
        &  & \(\phi_1\) & -0.41 & 1.08 & 1.07 & 0.96 & -0.05 & 1.56 & 1.49 & 0.96 \\
        &  & \(\psi_1\) & -0.49 & 0.91 & 0.94 & 0.94 & 0.07 & 1.57 & 1.64 & 0.94 \\
        & 200 & \(\alpha_1\) & -0.43 & 0.73 & 0.71 & 0.96 & -0.01 & 0.75 & 0.72 & 0.96 \\
        &  & \(\beta_1\) & -0.04 & 0.54 & 0.56 & 0.94 & -0.07 & 0.57 & 0.56 & 0.96 \\
        &  & \(\omega\) & 0.00 & 0.06 & 0.06 & 0.95 & 0.02 & 0.10 & 0.11 & 0.95 \\
        &  & \(\phi_1\) & -0.30 & 0.77 & 0.76 & 0.96 & -0.69 & 1.11 & 1.11 & 0.95 \\
        &  & \(\psi_1\) & -0.01 & 0.64 & 0.64 & 0.96 & -0.02 & 1.12 & 1.19 & 0.95 \\
        & 400 & \(\alpha_1\) & 0.05 & 0.52 & 0.51 & 0.95 & 0.30 & 0.53 & 0.54 & 0.94 \\
        &  & \(\beta_1\) & 0.14 & 0.38 & 0.36 & 0.96 & -0.13 & 0.40 & 0.41 & 0.95 \\
        &  & \(\omega\) & 0.02 & 0.05 & 0.05 & 0.95 & 0.05 & 0.07 & 0.08 & 0.95 \\
        &  & \(\phi_1\) & 0.01 & 0.54 & 0.55 & 0.95 & -0.05 & 0.81 & 0.79 & 0.96 \\
        &  & \(\psi_1\) & -0.19 & 0.45 & 0.47 & 0.94 & -0.07 & 0.81 & 0.86 & 0.95 \\
        \bottomrule
    \end{tabular*}
    \label{tab:sim_bias_DGP1_case1}
\end{table}

\begin{table}[h]
    \centering
    \caption{Simulation results for DGP 1 in Example \ref{ex:2}.
        The table reports the Biases(\(\x10^3\)), ASDs(\(\x10^2\)), ESDs(\(\x10^2\)),
        and CPs of the QMLE, where \( \eta_{it} \) follow the standard normal distribution
        or the scaled \( t_5 \) distribution.
        The network densities corresponding to each \( N \) are provided in parentheses below.}
    \begin{tabular*}{\textwidth}{@{\extracolsep{\fill}}cccrcccrccc}
        \toprule
        &     &              & \multicolumn{4}{c}{Normal} & \multicolumn{4}{c}{\( t_5 \)}                                           \\
        \cmidrule(lr){4-7} \cmidrule(lr){8-11}
        \( N \)      & \( T \)   &    QMLE      & Bias                       & ASD                           & ESD  & CP   & Bias & ASD  & ESD  & CP   \\
        \midrule
        50 & 100 & \(\alpha_1\) & -0.38 & 2.11 & 2.12 & 0.95 & -0.95 & 2.16 & 2.19 & 0.95 \\
        (5.67\%)&  & \(\beta_1\) & -0.53 & 1.51 & 1.51 & 0.95 & -0.67 & 1.60 & 1.57 & 0.96 \\
        &  & \(\omega\) & 0.15 & 0.18 & 0.19 & 0.94 & 0.06 & 0.28 & 0.29 & 0.94 \\
        &  & \(\phi_1\) & -0.83 & 2.20 & 2.24 & 0.94 & -0.81 & 2.99 & 2.98 & 0.94 \\
        &  & \(\psi_1\) & -2.20 & 1.80 & 1.92 & 0.92 & 0.07 & 3.01 & 3.21 & 0.94 \\
        & 200 & \(\alpha_1\) & -0.90 & 1.49 & 1.56 & 0.94 & -0.22 & 1.53 & 1.55 & 0.95 \\
        &  & \(\beta_1\) & -0.36 & 1.07 & 1.11 & 0.94 & -0.08 & 1.13 & 1.13 & 0.95 \\
        &  & \(\omega\) & 0.06 & 0.13 & 0.13 & 0.94 & -0.02 & 0.20 & 0.20 & 0.94 \\
        &  & \(\phi_1\) & -0.97 & 1.55 & 1.54 & 0.95 & -0.09 & 2.17 & 2.53 & 0.94 \\
        &  & \(\psi_1\) & -0.80 & 1.28 & 1.28 & 0.94 & 0.54 & 2.19 & 2.52 & 0.94 \\
        & 400 & \(\alpha_1\) & 0.31 & 1.06 & 1.05 & 0.95 & 0.24 & 1.09 & 1.08 & 0.95 \\
        &  & \(\beta_1\) & 0.01 & 0.76 & 0.74 & 0.96 & 0.05 & 0.80 & 0.79 & 0.94 \\
        &  & \(\omega\) & 0.08 & 0.09 & 0.09 & 0.95 & 0.01 & 0.14 & 0.15 & 0.95 \\
        &  & \(\phi_1\) & -0.70 & 1.10 & 1.12 & 0.94 & 0.31 & 1.56 & 1.65 & 0.95 \\
        &  & \(\psi_1\) & -0.28 & 0.91 & 0.90 & 0.94 & 0.04 & 1.56 & 1.63 & 0.95 \\
        200 & 100 & \(\alpha_1\) & 0.07 & 1.05 & 1.08 & 0.94 & -0.09 & 1.08 & 1.07 & 0.95 \\
        (1.44\%)&  & \(\beta_1\) & 0.49 & 0.76 & 0.76 & 0.94 & 0.23 & 0.80 & 0.80 & 0.95 \\
        &  & \(\omega\) & 0.04 & 0.09 & 0.09 & 0.94 & 0.03 & 0.14 & 0.15 & 0.96 \\
        &  & \(\phi_1\) & -0.47 & 1.09 & 1.13 & 0.95 & -1.10 & 1.56 & 1.50 & 0.94 \\
        &  & \(\psi_1\) & -0.47 & 0.91 & 0.90 & 0.96 & -0.28 & 1.56 & 1.58 & 0.94 \\
        & 200 & \(\alpha_1\) & 0.29 & 0.74 & 0.72 & 0.96 & -0.07 & 0.76 & 0.76 & 0.96 \\
        &  & \(\beta_1\) & 0.11 & 0.54 & 0.53 & 0.95 & 0.11 & 0.57 & 0.56 & 0.95 \\
        &  & \(\omega\) & 0.03 & 0.06 & 0.06 & 0.96 & 0.00 & 0.10 & 0.11 & 0.94 \\
        &  & \(\phi_1\) & -0.31 & 0.77 & 0.77 & 0.95 & 0.12 & 1.12 & 1.17 & 0.94 \\
        &  & \(\psi_1\) & -0.02 & 0.64 & 0.67 & 0.94 & -0.72 & 1.11 & 1.18 & 0.94 \\
        & 400 & \(\alpha_1\) & 0.00 & 0.53 & 0.54 & 0.94 & -0.15 & 0.54 & 0.54 & 0.94 \\
        &  & \(\beta_1\) & -0.10 & 0.38 & 0.36 & 0.96 & -0.12 & 0.40 & 0.39 & 0.95 \\
        &  & \(\omega\) & 0.01 & 0.05 & 0.04 & 0.96 & 0.04 & 0.07 & 0.08 & 0.95 \\
        &  & \(\phi_1\) & 0.18 & 0.55 & 0.56 & 0.94 & 0.07 & 0.81 & 0.76 & 0.96 \\
        &  & \(\psi_1\) & -0.14 & 0.45 & 0.45 & 0.95 & -0.60 & 0.80 & 0.79 & 0.95 \\
        \bottomrule
    \end{tabular*}
    \label{tab:sim_bias_DGP1_case2}
\end{table}

\begin{table}[h]
    \centering
    \caption{Simulation results of QMLE for DGP 1 in Example \ref{ex:3}.
        The table reports the Biases(\(\x10^3\)), ASDs(\(\x10^2\)), ESDs(\(\x10^2\)),
        and CPs of the QMLE, where \( \eta_{it} \) follow the standard normal distribution
        or the scaled \( t_5 \) distribution.
        The network densities corresponding to each \( N \) are provided in parentheses below.}
    \begin{tabular*}{\textwidth}{@{\extracolsep{\fill}}cccrcccrccc}
        \toprule
        &     &              & \multicolumn{4}{c}{Normal} & \multicolumn{4}{c}{\( t_5 \)}                                           \\
        \cmidrule(lr){4-7} \cmidrule(lr){8-11}
        \( N \)      & \( T \)   &    QMLE      & Bias                       & ASD                           & ESD  & CP   & Bias & ASD  & ESD  & CP   \\
        \midrule
        50 & 100 & \(\alpha_1\) & 0.01 & 2.05 & 2.10 & 0.95 & -0.10 & 2.11 & 2.11 & 0.95 \\
        (5.63\%)   &  & \(\beta_1\) & -0.56 & 1.51 & 1.45 & 0.96 & 0.30 & 1.60 & 1.65 & 0.94 \\
        &  & \(\omega\) & -0.04 & 0.18 & 0.18 & 0.94 & -0.04 & 0.27 & 0.28 & 0.95 \\
        &  & \(\phi_1\) & 0.20 & 2.13 & 2.14 & 0.94 & 0.35 & 2.91 & 2.98 & 0.92 \\
        &  & \(\psi_1\) & -1.14 & 1.81 & 1.85 & 0.94 & 0.64 & 2.99 & 3.15 & 0.94 \\
        & 200 & \(\alpha_1\) & -0.03 & 1.45 & 1.51 & 0.94 & -0.59 & 1.49 & 1.48 & 0.96 \\
        &  & \(\beta_1\) & 0.37 & 1.07 & 1.08 & 0.94 & 0.08 & 1.13 & 1.15 & 0.94 \\
        &  & \(\omega\) & -0.00 & 0.13 & 0.13 & 0.94 & -0.01 & 0.19 & 0.19 & 0.96 \\
        &  & \(\phi_1\) & -0.35 & 1.51 & 1.52 & 0.95 & -0.01 & 2.12 & 2.23 & 0.94 \\
        &  & \(\psi_1\) & -0.64 & 1.28 & 1.30 & 0.94 & 1.13 & 2.18 & 2.50 & 0.96 \\
        & 400 & \(\alpha_1\) & 0.07 & 1.03 & 0.98 & 0.96 & 0.23 & 1.06 & 1.07 & 0.95 \\
        &  & \(\beta_1\) & 0.36 & 0.76 & 0.76 & 0.95 & -0.28 & 0.80 & 0.80 & 0.96 \\
        &  & \(\omega\) & 0.05 & 0.09 & 0.09 & 0.95 & 0.02 & 0.14 & 0.14 & 0.95 \\
        &  & \(\phi_1\) & -0.52 & 1.06 & 1.06 & 0.95 & -0.52 & 1.53 & 1.60 & 0.94 \\
        &  & \(\psi_1\) & -0.31 & 0.91 & 0.90 & 0.94 & -0.39 & 1.55 & 1.53 & 0.95 \\
        200 & 100 & \(\alpha_1\) & 1.10 & 1.09 & 1.07 & 0.95 & -0.37 & 1.12 & 1.09 & 0.96 \\
        (1.43\%)   &  & \(\beta_1\) & -0.08 & 0.76 & 0.76 & 0.95 & -0.36 & 0.80 & 0.80 & 0.95 \\
        &  & \(\omega\) & -0.00 & 0.09 & 0.09 & 0.96 & -0.06 & 0.14 & 0.14 & 0.95 \\
        &  & \(\phi_1\) & 0.13 & 1.10 & 1.10 & 0.95 & -0.11 & 1.58 & 1.48 & 0.96 \\
        &  & \(\psi_1\) & -0.42 & 0.91 & 0.92 & 0.95 & 0.81 & 1.56 & 2.10 & 0.93 \\
        & 200 & \(\alpha_1\) & -0.08 & 0.77 & 0.79 & 0.95 & -0.15 & 0.79 & 0.78 & 0.95 \\
        &  & \(\beta_1\) & 0.02 & 0.54 & 0.54 & 0.93 & -0.03 & 0.57 & 0.58 & 0.94 \\
        &  & \(\omega\) & 0.03 & 0.06 & 0.06 & 0.95 & 0.03 & 0.10 & 0.10 & 0.94 \\
        &  & \(\phi_1\) & 0.01 & 0.78 & 0.78 & 0.95 & 0.05 & 1.15 & 1.16 & 0.96 \\
        &  & \(\psi_1\) & -0.30 & 0.64 & 0.64 & 0.95 & -1.07 & 1.12 & 1.12 & 0.94 \\
        & 400 & \(\alpha_1\) & -0.07 & 0.54 & 0.54 & 0.95 & -0.02 & 0.56 & 0.56 & 0.95 \\
        &  & \(\beta_1\) & 0.06 & 0.38 & 0.38 & 0.96 & -0.43 & 0.40 & 0.42 & 0.93 \\
        &  & \(\omega\) & -0.01 & 0.04 & 0.04 & 0.96 & 0.02 & 0.07 & 0.07 & 0.96 \\
        &  & \(\phi_1\) & 0.03 & 0.55 & 0.56 & 0.96 & -0.20 & 0.81 & 0.82 & 0.95 \\
        &  & \(\psi_1\) & 0.35 & 0.45 & 0.44 & 0.96 & 0.12 & 0.80 & 0.82 & 0.95 \\
        \bottomrule
    \end{tabular*}
    \label{tab:sim_bias_DGP1_case3}
\end{table}

In DGP 2, to observe the BIC performance, different order between network-autoregression terms and autoregression terms are set.
Let \( r_m=3 \) and obtain the \( (\hat{p},\hat{q}) \) by solving \eqref{eq:bic} for each replication.
The cases of underfitting, correct selection, and overfitting by BIC correspond
to \( \{\hat{p}<1 \text{ or } \hat{q}<2\} \), \( \{\hat{p}=1, \hat{q}=2\} \) 
and \( \{\hat{p}\geq1,\hat{q}>2\}\cup\{\hat{p}>1,\hat{q}\geq2\} \) respectively.
Table \ref{tab:sim_BIC} reports the frequencies of the three cases for different sample sizes and distributions. 
The performance of BIC is slightly more satisfactory 
when the network structure follows the randomly distributed structure compared to the other two cases. 
While \( \{\eta_{it}\} \) follows a normal distribution, 
it is more likely to select a lower order than the scaled \( t_5 \) distribution when \( T=100 \); 
however, the former exhibits a faster rate than the latter. 
The heavy tails of the \( t_5 \) distribution make BIC more likely to select higher-order models,
these phinominen may contribute to that better results when \( T=100 \). 
The number of nodes \( N \) has an important effect on the performance of BIC, especially when the time span is short.

\begin{table}[htb]
    \centering
    \caption{Simulation results of BIC for DGP 2.
        The table reports the frequencies of lower orders, correct orders and higher orders by BIC,
        where \( \eta_{it} \) follow the standard normal distribution or the scaled \( t_5 \) distribution.}
    \begin{tabular}{@{\extracolsep{\fill}}cccrrrrrrr}
        \toprule
            &      &         &     & \multicolumn{3}{c}{Normal} & \multicolumn{3}{c}{\( t_{5} \)}                                   \\
        \cmidrule(lr){5-7} \cmidrule(lr){8-10}
        N   & Case & Density & T   & Lower                      & Exact                           & Higher & Lower & Exact & Higher \\
        \midrule
        50  & 1    & 5.80\%  & 100 & 304                        & 695                             & 1      & 180   & 771   & 49     \\
            &      &         & 200 & 24                         & 972                             & 4      & 6     & 935   & 59     \\
            &      &         & 300 & 3                          & 996                             & 1      & 0     & 957   & 43     \\
            & 2    & 5.67\%  & 100 & 345                        & 652                             & 3      & 219   & 727   & 54     \\
            &      &         & 200 & 37                         & 962                             & 1      & 20    & 927   & 53     \\
            &      &         & 300 & 1                          & 999                             & 0      & 1     & 951   & 48     \\
            & 3    & 5.63\%  & 100 & 321                        & 672                             & 7      & 199   & 741   & 60     \\
            &      &         & 200 & 36                         & 961                             & 3      & 14    & 940   & 46     \\
            &      &         & 300 & 2                          & 992                             & 6      & 1     & 953   & 46     \\
        200 & 1    & 1.48\%  & 100 & 0                          & 994                             & 6      & 0     & 932   & 68     \\
            &      &         & 200 & 0                          & 996                             & 4      & 0     & 953   & 47     \\
            &      &         & 300 & 0                          & 995                             & 5      & 0     & 956   & 44     \\
            & 2    & 1.44\%  & 100 & 0                          & 990                             & 10     & 0     & 937   & 63     \\
            &      &         & 200 & 0                          & 998                             & 2      & 0     & 961   & 39     \\
            &      &         & 300 & 0                          & 997                             & 3      & 0     & 966   & 34     \\
            & 3    & 1.43\%  & 100 & 0                          & 996                             & 4      & 0     & 939   & 61     \\
            &      &         & 200 & 0                          & 999                             & 1      & 0     & 933   & 67     \\
            &      &         & 300 & 0                          & 998                             & 2      & 0     & 958   & 42     \\
        \bottomrule
    \end{tabular}
    \label{tab:sim_BIC}
\end{table}

\section{Applications}

In this section, we illustrate the performance of the proposed NDAR model on a real dataset. 
We collected stock data from the A-share markets of the Shanghai Stock Exchange and 
the Shenzhen Stock Exchange across three different industries: 
manufacturing (61 stocks), wholesale and retail (43 stocks), 
and information transmission, software, and information technology services (abbreviated to IT, 68 stocks). 
The dataset spans from January 3, 2023, to July 1, 2024, with a total of \( T = 359 \) trading days. 
For each stock \( i \), we computed the daily logarithmic return as the observed variable \( y_{it} \). 

The three categories of stocks were analyzed separately. For each category, the network structure 
was defined based on whether there were common shareholders among the top ten holders of the stocks. 
In practice, certain entities, such as Hong Kong Securities Clearing Company Ltd., are commonly 
listed as shareholders for many stocks. Including only one common shareholder would misrepresent 
the actual ownership structure, resulting in a very dense network and leading to unstable parameter 
estimates in the model. 
To address this issue, we only considered stocks with at least two common major shareholders among 
their top ten holders. Specifically, the network was constructed such that \( a_{ij} = 1 \) if there 
were at least two common holders between stock \( i \) and stock \( j \) within the top ten shareholders.

The average log return and the network structure for each category are presented. 
Figures \ref{fig:return_group_C} and \ref{fig:return_group_F} show that 
the average log returns of manufacturing and wholesale and retail stocks exhibit similar behavior, 
but the former demonstrates two large communities, while the latter consists of a single large network with a few mutually connected stocks. 
The IT and Software Services category exhibits higher volatility compared to the other two categories overall, 
with its network structure shown in Figure \ref{fig:return_group_I}.

The proposed model \ref{eq:NDAR_pq} was applied to fit the observed data for each category of stocks, 
with the model order selected based on the BIC defined in \eqref{eq:bic}.
After setting a maximum order of \( r_m = 5 \),
the orders selected for the three categories all were \( (1, 3) \). 
Table \ref{tab:real_data} presents the estimated parameters for each category of stocks. 
Statistical tests were conducted under the null hypothesis that the true parameter values are 0. 
For parameters in the mean parameters, two-sided tests were used, while for volatility parameters, 
one-sided tests were employed, then the p-values are calculated.

The results demonstrate that network structure influences stock returns to varying degrees across categories. 
In the manufacturing sector, the network effect terms in the mean (\( \alpha_1 = 0.0358 \)) and volatility 
(\( \phi_1 = 0.095 \)) components were both statistically significant, suggesting that ownership linkages 
affect the conditional expectation and variability of returns. For the wholesale and retail sector, the network 
term in the mean component was not significant (\( p\text{-value} = 0.616 \)), implying that autoregressive effects 
primarily drive the conditional expectation of returns. In the IT sector, significant network terms were found 
in both the mean (\( 0.186 \)) and volatility (\( 0.0504 \)) components, with the baseline variance (\( \omega \)) 
notably higher than in other sectors. Furthermore, the autoregressive terms in the volatility component were 
positive and significant across all sectors, confirming the persistence of return volatility.
The heterogeneity in network effects and suggest that the influence of network structures
varies depending on stock categories.

\begin{figure}
    \centering
    \begin{subfigure}{0.5\textwidth}
        \includegraphics[width=\textwidth]{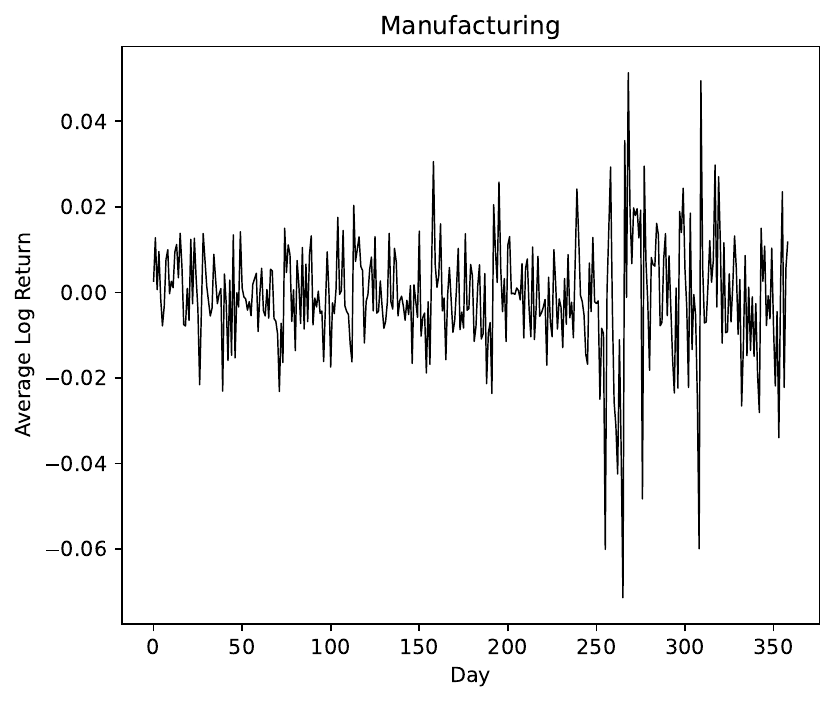}
        \caption{}
    \end{subfigure}
    \begin{subfigure}{0.42\textwidth}
        \includegraphics[width=\textwidth]{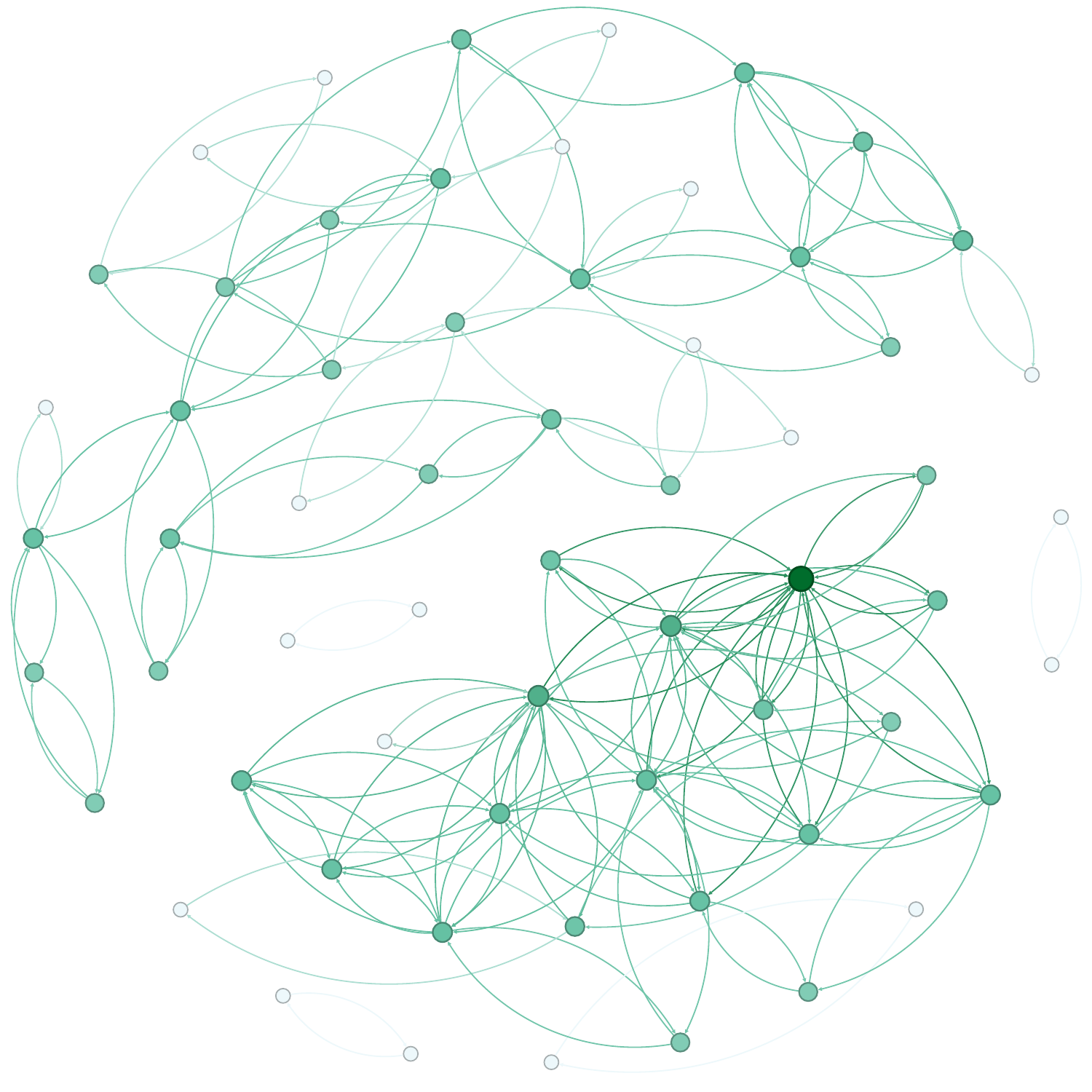}
        \caption{}
    \end{subfigure}
    \caption{Log returns and network structure of 61 manufacturing stocks. 
    (a) Average daily log returns, with January 3, 2023, as the time origin. 
    (b) Network matrix with a density of 4.75\%, where darker-colored nodes indicate higher in-degrees.}    
    \label{fig:return_group_C}
\end{figure}

\begin{figure}[htbp]
    \centering
    \begin{subfigure}{0.5\textwidth}
        \includegraphics[width=\textwidth]{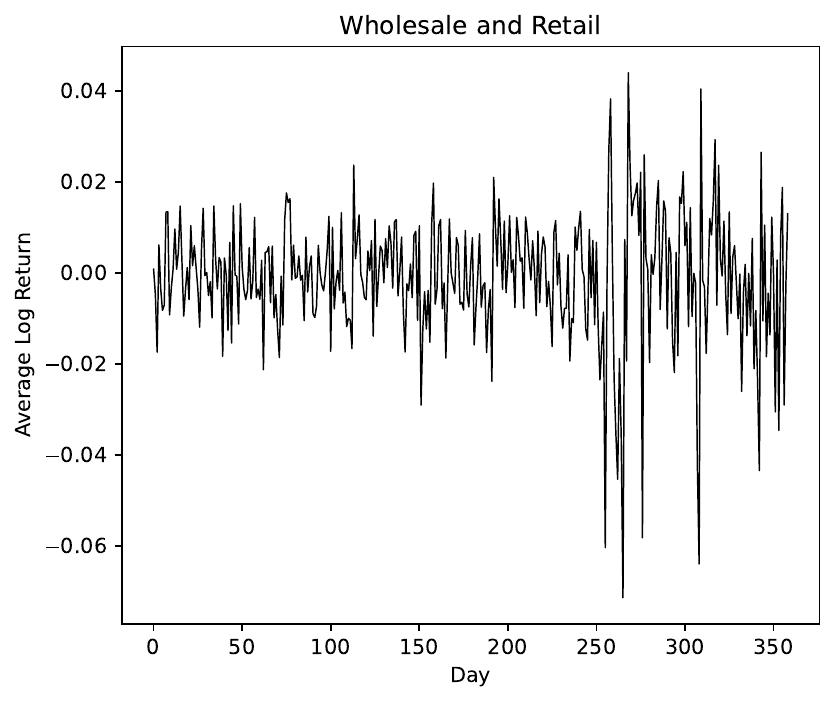}
        \caption{}
    \end{subfigure}
    \begin{subfigure}{0.42\textwidth}
        \includegraphics[width=\textwidth]{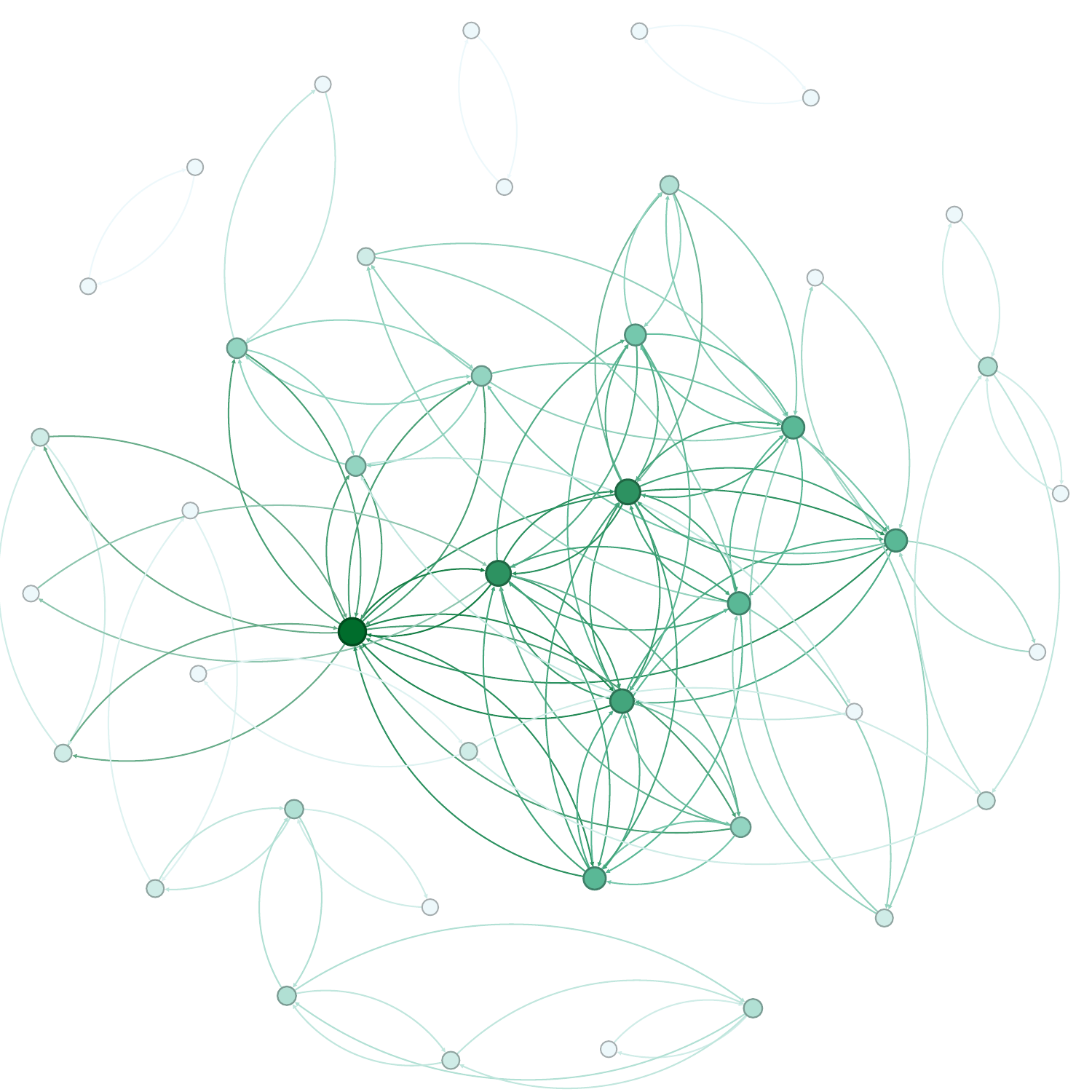}
        \caption{}
    \end{subfigure}
    \caption{Log returns and network structure of 43 wholesale and retail stocks. 
    (a) Average daily log returns, with January 3, 2023, as the time origin. 
    (b) Network matrix with a density of 6.98\%, where darker-colored nodes indicate higher in-degrees.}    
    \label{fig:return_group_F}
\end{figure}

\begin{figure}[htbp]
    \centering
    \begin{subfigure}{0.5\textwidth}
        \includegraphics[width=\textwidth]{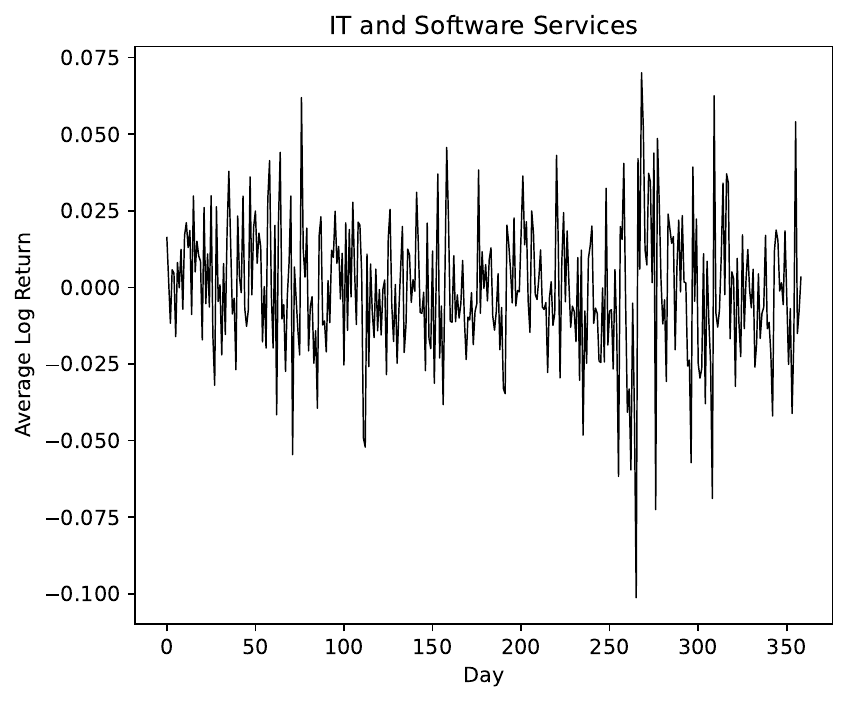}
        \caption{}
    \end{subfigure}
    \begin{subfigure}{0.42\textwidth}
        \includegraphics[width=\textwidth]{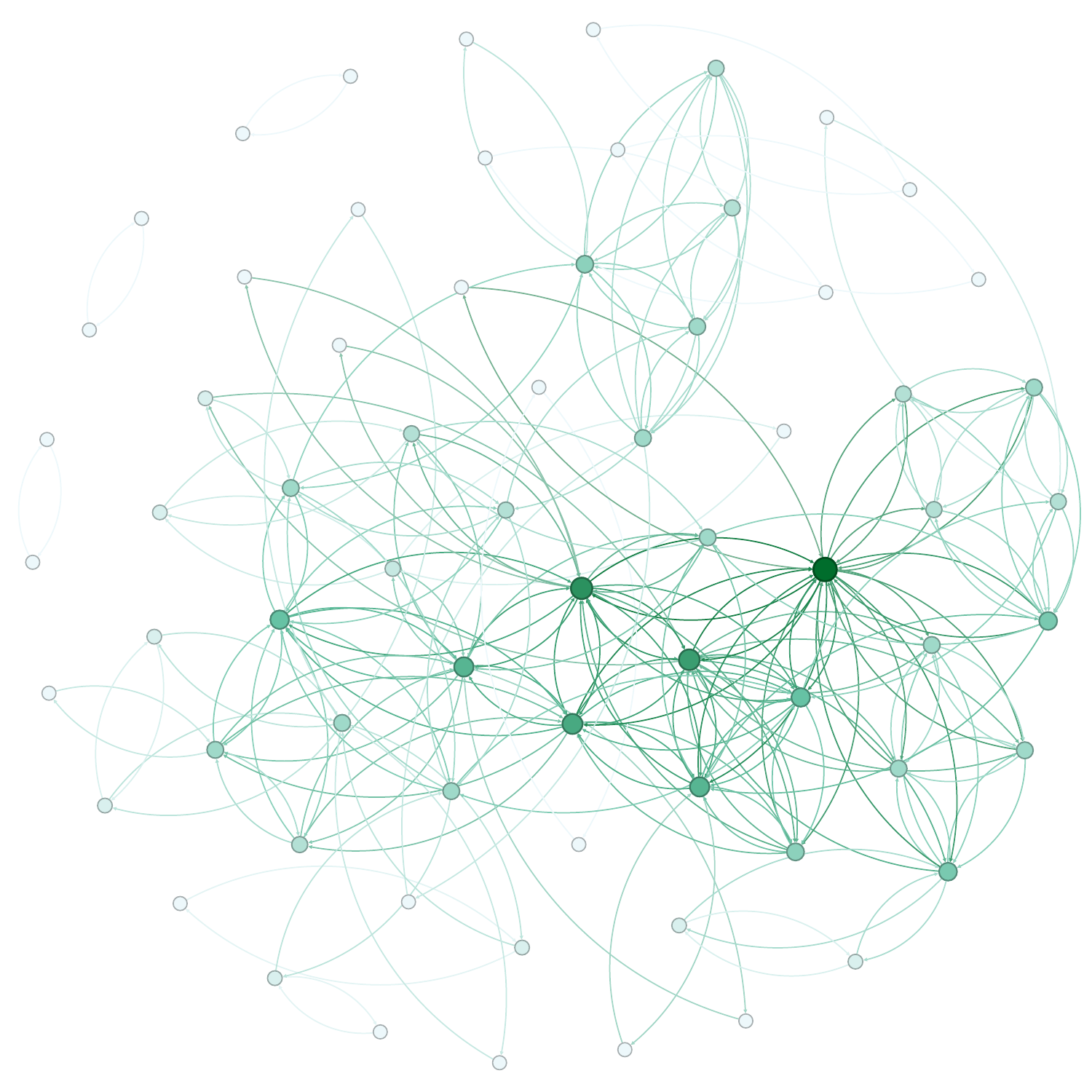}
        \caption{}
    \end{subfigure}
    \caption{Log returns and network structure of 68 IT stocks. 
    (a) Average daily log returns, with January 3, 2023, as the time origin. 
    (b) Network matrix with a density of 5.31\%, where darker-colored nodes indicate higher in-degrees.}    
    \label{fig:return_group_I}
\end{figure}

\begin{table}[htb]
    \centering
    \caption{The proposed NDAR model was applied to real stock data from three categories. 
    The model order for each category was determined to be (1, 3) based on the BIC. 
    Standard errors are provided in parentheses.}    
    \begin{tabular}{crrrrrrrrrr}
        \toprule
                           & \multicolumn{2}{c}{Manufacturing} & \multicolumn{2}{c}{Wholesale and Retail} & \multicolumn{2}{c}{IT}                                                 \\
        \cmidrule(lr){2-3} \cmidrule(lr){4-5} \cmidrule(lr){6-7}
        Parameter             & Estimate                      & \( p \)-value                            & Estimate                                & \( p \)-value & Estimate  & \( p \)-value \\
        \midrule
        \( \alpha_1 \) & .0358(.0097)                     & <.001                                    & .0049(.0097)                                & .616          & .0186(.0087) & .034          \\
        \( \beta_1 \)  & .0215(.0096)                     & .025                                     & .0246(.0111)                                & .026          & .0073(.0093) & .433          \\
        \( \beta_2 \)  & .0259(.0087)                     & .028                                     & .0715(.0103)                                & <.001         & .0209(.0079) & .008          \\
        \( \beta_3 \)  & .0151(.0086)                     & .080                                     & .0404(.0104)                                & <.001         & .0080(.0078) & .304          \\
        \( \omega \)   & .0002(.0000)                     & <.001                                    & .0002(.0000)                                & <.001         & .0006(.0000) & <.001         \\
        \( \phi_1 \)   & .0950(.0276)                     & <.001                                    & .0350(.0147)                                & .009          & .0504(.0152) & <.001         \\
        \( \psi_1 \)   & .3002(.0466)                     & <.001                                    & .2500(.0361)                                & <.001         & .2519(.0277) & <.001         \\
        \( \psi_2 \)   & .2030(.0395)                     & <.001                                    & .2226(.0345)                                & <.001         & .1792(.0238) & <.001         \\
        \( \psi_3 \)   & .2028(.0384)                     & <.001                                    & .2649(.0362)                                & <.001         & .1733(.0230) & <.001         \\
        \bottomrule
    \end{tabular}
    \label{tab:real_data}
\end{table}

\section{Discussion}

In this paper, we propose a network double autoregressive model
to describe the dynamic evolution of multi time series data. 
We apply a quasi-maximum likelihood estimation method to estimate model parameters, 
and establish the consistency and asymptotic normality of the estimators. 
To address the model selection problem with unknown lag orders, 
we propose a Bayesian Information Criterion based method for selecting the lag order. 
Simulation studies demonstrate that the proposed method performs well in finite samples. 
Furthermore, we apply our method to three categories of stock data, 
constructing network relationships based on common shareholders. 
Statistical inference shows that network effects may influence the conditional expectation, 
and variance of stock returns.

There are several potential directions to extend this model.
First, while the Bayesian Information Criterion is used for model selection, 
the risk of overfitting still exists
since redundant parameters located in non-ending positions may still be selected. 
Future research could explore theories for cases where the true values 
of some parameters in the volatility component are zero, enabling richer statistical inference. 
Second, in our model, the heterogeneity of observations across different nodes 
is determined solely by their network topological structure, without taking into account the heterogeneity of parameters.
When the data involve more nodes, it may be more reasonable in practical applications 
to classify nodes and allow different nodes to have distinct parameters. 
Finally, our model is based on time series data, but in real-world applications, 
considering node-specific covariates could provide additional information to model data. 

\bibliographystyle{plainnat}
\bibliography{references}

\appendix  %
\renewcommand{\theequation}{A\arabic{equation}}  %
\renewcommand{\thesection}{A\arabic{section}}
\renewcommand{\thetheorem}{A\arabic{theorem}}

\section*{Appendix}

\subsection*{Proof of Proposition~\ref{prop:strictly_stationary}}

This proposition follows directly from Theorem 1 in \cite{lu2001l1}. 
For this proposition, without loss of generality, 
we assume the model \eqref{eq:NDAR_pq} has order \( (p, p) \)
with additional parameters set to zero.
The model \eqref{eq:NDAR_pq} can be reformulated as:
\[
    \bm{y}_t = f(\bm{y}_{t-1}, \ldots, \bm{y}_{t-p}) + H(\bm{y}_{t-1}, \ldots, \bm{y}_{t-p})\bm{\eta}_t,
\]
where 
\[
    f = \sum_{r=1}^{p} (\alpha W + \beta I) \bm{y}_{t-r},
    \quad
    H^2 = \omega I + \mathrm{Diag}\left\{\sum_{r=1}^{p}(\phi W + \psi I) \bm{x}_{t-r}\right\}.
\]
The result in Proposition~\ref{prop:strictly_stationary} follows directly from the inequality conditions (B2) in \cite{lu2001l1}.

\subsection*{Proof of Theorem~\ref{theorem:normal}}

Let \( \varepsilon_{it}(\mu) = y_{it} - \bz_{1it}'\mu \) and \( h_{it}(\sigma) = \bz_{2it}'\sigma \). 
Then
\[
    \varepsilon_{it}\left( \mu \right)
    = \varepsilon_{it}\left( \mu_0 \right)
    + \left( \mu_0-\mu \right)' \bz_{1it}
    = \eta_{it}\sqrt{\sigma_0' \bz_{2it}} + \left( \mu_0-\mu \right)' \bz_{1it}
\]
and the log-likelihood contribution is
\[
    l_{it}\left( \theta \right) = -\frac{1}{2}\left( \log{h_{it}(\sigma)} + \frac{\varepsilon_{it}^{2}(\mu)}{h_{it}(\sigma)} \right)
\]
Here, we list derivatives of $ l_{it}\left( \theta \right)$:
\begin{equation*}
    \begin{aligned}
        \frac{\partial l_{it}(\theta)}{\partial \mu} & = \frac{\varepsilon_{it}(\mu)}{h_{it}(\sigma)}\bz_{1it}, \\
        \frac{\partial l_{it}(\theta)}{\partial \sigma} & = \frac{1}{2h_{it}(\sigma)}\left( \frac{\varepsilon_{it}^{2}(\mu)}{h_{it}(\sigma)}-1 \right)\bz_{2it}.
    \end{aligned}
\end{equation*}
Next, we have:
\begin{equation*}
    \begin{aligned}
        \frac{\partial l_{it}(\theta)}{\partial \mu} \frac{\partial l_{it}(\theta)}{\partial \theta'_{1}} & = \frac{\varepsilon_{it}^{2}(\mu)}{h_{it}^{2}(\sigma)} \bz_{1it}\bz'_{1it}, \\
        \frac{\partial l_{it}(\theta)}{\partial \mu} \frac{\partial l_{it}(\theta)}{\partial \theta'_{2}} & = \frac{1}{2}\frac{\varepsilon_{it}(\mu)}{h^{2}_{it}(\sigma)} \left( \frac{\varepsilon^{2}_{it}(\mu)}{h_{it}(\sigma)} - 1 \right) \bz_{1it}\bz'_{2it}, \\
        \frac{\partial l_{it}(\theta)}{\partial \sigma} \frac{\partial l_{it}(\theta)}{\partial \theta'_{2}} & = \frac{1}{4}\frac{1}{h^{2}_{it}(\sigma)} \left( \frac{\varepsilon^{2}_{it}(\mu)}{h_{it}(\sigma)} - 1 \right)^{2} \bz_{2it}\bz'_{2it}.
    \end{aligned}
\end{equation*}
The second derivatives are
\begin{equation*}
    \begin{aligned}
        \frac{\partial^{2} l_{it}(\theta)}{\partial \mu \partial \mu'} & = -\frac{1}{h_{it}(\sigma)}\bz_{1it}\bz'_{1it}, \\
        \frac{\partial^{2} l_{it}(\theta)}{\partial \mu \partial \sigma'} & = -\frac{\varepsilon_{it}(\mu)}{h_{it}^{2}(\sigma)} \bz_{1it}\bz'_{2it}, \\
        \frac{\partial^{2} l_{it}(\theta)}{\partial \sigma \partial \sigma'} & = \frac{1}{h_{it}^{2}(\sigma)}\left( \frac{1}{2} -\frac{\varepsilon_{it}^{2}(\mu)}{h_{it}(\sigma)} \right) \bz_{2it}\bz'_{2it}.
    \end{aligned}
\end{equation*}

Recall that
\begin{equation*}
    \everymath{\displaystyle}
    D = \begin{pmatrix}
        1 & \frac{\kappa_3}{\sqrt{2}} \\ \frac{\kappa_3}{\sqrt{2}} & \frac{\kappa_4-1}{2}
    \end{pmatrix}
    ~
    \varGamma_{it}(\theta) = \Diag{\frac{\bz_{1it}}{\sqrt{h_{it}(\sigma)}}, \frac{\bz_{2it}}{\sqrt{2}h_{it}(\sigma)}}
\end{equation*}
and 
\begin{equation*}
    \begin{aligned}
        \Sigma_0 &\equiv E \left[ \frac{\partial l_{it}(\theta_0)}{\partial \theta} \frac{\partial l_{it}(\theta_0)}{\partial \theta'} \right] = E \left[ \varGamma_{it}(\theta_0) D  \varGamma'_{it}(\theta_0) \right]
        \\
        \Omega_0 &\equiv -E \left[ \frac{\partial^{2} l_{it}(\theta_0)}{\partial \theta \partial \theta'} \right] = E \left[ \varGamma_{it}(\theta_0)  \varGamma'_{it}(\theta_0) \right]
        \\
    \end{aligned}
\end{equation*}
Denote the element-wise bounds of \( \mu \) and \( \sigma \) as
\[
    \bar{\mu} = \bar{b}_1 \bm{1}_{p+q},\quad \bar{\sigma} = \bar{b}_2\bm{1}_{p+q+1},\quad \text{and} \quad \underline{\sigma} = \underline{b}_2\bm{1}_{p+q+1}
\]

We need to prove several lemmas first, and the proof of the theorem is given at the end of this subsection.

\begin{lemma}\label{lemma1}
    Let \( g(\by, \theta) \) be a measurable function of \( y \) in Euclidean space for each \( \theta\in\Theta \),
    a compact subset of \( \R^{p} \), and a continuous function of \( \theta\in\Theta \) for each \( \by \).
    Suppose that \( \left\{ {\by_t} \right\}  \) is a sequence of strictly stationary
    and ergodic time series such that 
    \( Eg(\by_t,\theta)=0 \) and \( E\sup_{\theta\in\Theta}|g(\by_t,\theta)|<\infty \).
    Then \( \sup_{\theta\in\Theta}|T^{-1}\sum_{t=T}^{n}g(\by_t,\theta)| = o_{p}(1) \).
\end{lemma}
\begin{proof}
    This is Theorem 3.1 in \cite{ling2003asymptotic}.
\end{proof}
\begin{lemma}\label{lemma2}
    If the conditions of theorem \ref{theorem:normal} hold, then
    \begin{enumerate}[label=(\roman*), font=\normalfont]
        \item $\displaystyle
                  E\sup_{\theta\in\Theta}\left\lvert L_{t}(\theta) \right\rvert < \infty ; $
        \item \(\displaystyle E\sup_{\theta\in\Theta} \left\lVert \frac{\partial L_{t}(\theta)}{\partial\theta} \right\rVert < \infty ; \)
        \item $\displaystyle
                  E\sup_{\theta\in\Theta} \left\lVert \frac{\partial^2 L_{t}(\theta)}{\partial\theta \partial\theta'} \right\rVert < \infty . $
    \end{enumerate}
\end{lemma}
\begin{proof}
    (i)
    Since \( E\left\lVert \by_t \right\rVert^u < \infty \),
    by the \( C_r \) inequality and Jensen's inequality, we have
    \begin{equation}\label{lemma2:ieq1}
        \begin{aligned}
            E \log h_{it}(\bar{\sigma}) 
            & = \frac{2}{u} E \log h_{it}^{\frac{u}{2}}(\bar{\sigma}) \\
            & \leq \frac{2}{u} E \log \left( \bar{b}_2^{*\frac{u}{2}} + \sum_{r=1}^{p} \bar{b}_2^{\frac{u}{2}} n_{i}^{-\frac{u}{2}} \sum_{j \neq i} a_{ij} |y_{j(t-r)}|^u + \sum_{r=1}^{q} \bar{b}_2^{\frac{u}{2}} |y_{i(t-r)}|^u \right) \\
            & \leq \frac{2}{u} \log \left( \bar{b}_2^{*\frac{u}{2}} + \sum_{r=1}^{p} \bar{b}_2^{\frac{u}{2}} n_{i}^{-\frac{u}{2}} \sum_{j \neq i} a_{ij} E|y_{j(t-r)}|^u + \sum_{r=1}^{q} \bar{b}_2^{\frac{u}{2}} E|y_{i(t-r)}|^u \right) \\
            & < \infty,
        \end{aligned}
    \end{equation}
    where \( \bar{b}_2^* = \max\{\bar{b}_2, 1\} \).
    By \eqref{lemma2:ieq1}, we obtain
    \begin{equation}\label{ieq:logh}
        \begin{aligned}
            E \sup_{\theta \in \Theta} |\log h_{it}(\sigma)| 
            & \leq E \sup_{\theta \in \Theta} \left( I_{\{h_{it}(\sigma) \geq 1\}} \log h_{it}(\sigma) \right) \\
            & \quad + E \sup_{\theta \in \Theta} \left( -I_{\{h_{it}(\sigma) < 1\}} \log h_{it}(\sigma) \right) \\
            & \leq E \log h_{it}(\bar{b}_2) - I_{\{\underline{b}_2 < 1\}} \log \underline{b}_2 \\
            & < \infty.
        \end{aligned}
    \end{equation}
    By norm inequalities, we derive some useful inequalities:
    \begin{align}
        &\left\lVert \bz_{2it} \right\rVert  \leq \left\lVert \bz_{2it} \right\rVert_{1}, \label{ieq:z2_norm2} \\
        &\left\lVert \bz_{1it} \right\rVert^2  \leq (p+q) \left\lVert \bz_{2it} \right\rVert_{1}, \label{ieq:z1_norm2}\\
        &E \sup_{\theta \in \Theta} \left\lVert \frac{\bz_{2it}}{h_{it}(\sigma)} \right\rVert_{1}^{k} 
         \leq E \frac{\left\lVert \bz_{2it} \right\rVert_{1}^k}{h_{it}^k(\underline{\sigma})} 
         = E \frac{\left\lVert \bz_{2it} \right\rVert_{1}^k}{\left\lVert \underline{\sigma}' \bz_{2it} \right\rVert_{1}^k} 
         = E \frac{\left\lVert \bz_{2it} \right\rVert_{1}^k}{\underline{b}_2 \left\lVert \bz_{2it} \right\rVert_{1}^k} 
         < \infty, \quad k = 1, 2, \ldots \label{ieq:z2_norm1} 
    \end{align}
    where \( \left\lVert \cdot \right\rVert_{p} \) denotes the \( l_p \) norm, 
    and the subscript is omitted when \( p=2 \).
    By the relation \( \varepsilon_{it}(\mu) = \varepsilon_{it}(\mu_0) - (\mu_0 - \mu)' \bz_{1it} \), it follows that
    \begin{equation}\label{ieq:e2_d_h}
        \begin{aligned}
            E \sup_{\theta \in \Theta} \left\lvert \frac{\varepsilon_{it}^2(\mu)}{h_{it}(\sigma)} \right\rvert
            & \leq 2E \sup_{\theta \in \Theta} \frac{\sigma'_0 \bz_{2it}}{h_{it}(\sigma)} 
            + 2E \sup_{\theta \in \Theta} \frac{\left( (\mu_0 - \mu)' \bz_{1it} \right)^2}{h_{it}(\sigma)} \\
            & \leq 2E \frac{\bar{\sigma}' \bz_{2it}}{\underline{\sigma}' \bz_{2it}} 
            + 2C E \sup_{\theta \in \Theta} \left\lVert \frac{\bz_{1it}}{h^{1/2}_{it}(\sigma)} \right\rVert^2 \\
            & \leq 2\frac{\bar{\sigma}}{\underline{\sigma}} + 2C E \sup_{\theta \in \Theta} \left\lVert \frac{\bz_{1it}}{h^{1/2}_{it}(\sigma)} \right\rVert^2 \\
            & < \infty.
        \end{aligned}
    \end{equation}
    Here, the last inequality holds due to 
    \eqref{ieq:z1_norm2} and \eqref{ieq:z2_norm1}.
    
    Finally, by \eqref{ieq:logh}, \eqref{ieq:e2_d_h}, and the fact that \( |L_t(\theta)| \leq \sum_{i=1}^{N} |l_{it}(\theta)| \),
    statement (i) can be proven.

    (ii) By \eqref{ieq:z2_norm2}, \eqref{ieq:z1_norm2} and \eqref{ieq:e2_d_h}, we have
    \begin{equation*}
        \begin{aligned}
            E\sup_{\theta\in\Theta}\left\lVert \frac{\partial l_{it}(\theta)}{\partial\mu} \right\rVert & = E\sup_{\theta\in\Theta} \left\lVert \bz_{1it}\frac{\varepsilon_{it}(\mu)}{h_{it}(\sigma)} \right\rVert \\
                                                                                                               & \leq E\sup_{\theta\in\Theta} \left\lvert \frac{\varepsilon_{it}(\mu)}{h^{1/2}_{it}(\sigma)} \right\rvert \left\lVert \frac{\bz_{1it}}{h^{1/2}_{it}(\sigma)} \right\rVert < \infty \\
        \end{aligned}
    \end{equation*}
    \begin{equation*}
        \begin{aligned}
            E\sup_{\theta\in\Theta}\left\lVert \frac{\partial l_{it}(\theta)}{\partial\sigma} \right\rVert & = E\sup_{\theta\in\Theta} \left\lVert \bz_{2it}\frac{1}{2h_{it}(\sigma)}\left( \frac{\varepsilon_{it}^{2}(\mu)}{h_{it}(\sigma)}-1 \right) \right\rVert \\
                                                                                                               & \leq E\sup_{\theta\in\Theta} \left\lVert \frac{\bz_{2it}}{2h_{it}(\sigma)}  \right\rVert \left\lvert \frac{\varepsilon_{it}^{2}(\mu)}{h_{it}(\sigma)}-1 \right\rvert  < \infty
        \end{aligned}
    \end{equation*}

    (iii) Similarly, by \eqref{ieq:z2_norm2}, \eqref{ieq:z1_norm2} and \eqref{ieq:e2_d_h}, we have
    \begin{equation*}
        \begin{aligned}
            E\sup_{\theta\in\Theta}\left\lVert \frac{\partial^2 l_{it}(\theta)}{\partial\mu \partial\theta'_{1}} \right\rVert   & = E\sup_{\theta\in\Theta}\left\lVert \bz_{1it}\bz'_{1it}/h_{it}(\sigma) \right\rVert                                                                                                                                                                \\
                                                                                                                                       & =  E\sup_{\theta\in\Theta} \left\lVert \bz_{1it} \right\rVert^2 /h_{it}(\sigma)                                                                                                                                                                          \\
                                                                                                                                       & < \infty,                                                                                                                                                                                                                                            \\
            E\sup_{\theta\in\Theta}\left\lVert \frac{\partial^2 l_{it}(\theta)}{\partial\mu \partial\theta'_{2}} \right\rVert^2 & = E\sup_{\theta\in\Theta}\left\lVert \bz_{1it}\frac{\varepsilon_{it}(\mu)}{h^2_{it}(\sigma)}\bz'_{2it} \right\rVert^2                                                                                                                           \\
                                                                                                                                       & \leq  E\sup_{\theta\in\Theta} \left\lvert \frac{\varepsilon_{it}(\mu)}{h_{it}^{1/2}(\sigma)} \right\rvert^2 \left\lVert \frac{\bz_{1it}}{h_{it}^{1/2}(\sigma)} \right\rVert^2 \left\lVert \frac{\bz_{2it}}{h_{it}(\sigma)} \right\rVert^2 \\
                                                                                                                                       & < \infty,                                                                                                                                                                                                                                            \\
            E\sup_{\theta\in\Theta}\left\lVert \frac{\partial^2 l_{it}(\theta)}{\partial\sigma \partial\theta'_{2}} \right\rVert^2   & = E\sup_{\theta\in\Theta}\left\lVert \bz_{2it}\frac{1}{h_{it}^{2}(\sigma)}\left( \frac{1}{2}-\frac{\varepsilon_{it}^{2}(\mu)}{h_{it}(\sigma)} \right)\bz'_{2it} \right\rVert^2                                                                                                      \\
                                                                                                                                       & \leq \frac{1}{2}E\sup_{\theta\in\Theta}\left\lVert \frac{\bz_{2it}}{h_{it}(\sigma)} \right\rVert^2 + E\sup_{\theta\in\Theta}\left\lvert \frac{\varepsilon_{it}^2(\mu)}{h_{it}(\sigma)} \right\rvert\left\lVert \frac{\bz_{2it}}{h_{it}(\sigma)} \right\rVert^2                                                                                                                                                                                                                                        \\
                                                                                                                                       & < \infty.                                                                                                                                                                                                                                            \\
        \end{aligned}
    \end{equation*}
\end{proof}

\begin{lemma}\label{lemma3}
    If the conditions of theorem \ref{theorem:normal} hold, then
    \begin{enumerate}[label=(\roman*), font=\normalfont]
        \item $\displaystyle
                  \sup_{\theta\in\Theta}\left|\frac{1}{T}\sum_{t=1}^{T} L_{t}(\theta) - E L_t(\theta) \right| = o_p(1) ;$
        \item $\displaystyle
                  \sup_{\theta\in\Theta}\left|\frac{1}{T}\sum_{t=1}^{T} \frac{\partial L_{t}(\theta)}{\partial \theta} - E \frac{\partial L_t(\theta)}{\partial \theta} \right| = o_p(1) ;$
        \item $\displaystyle
                  \sup_{\theta\in\Theta}\left\|\frac{1}{T}\sum_{t=1}^{T} \frac{\partial^2 L_t(\theta)}{\partial\theta \partial\theta'} - E \frac{\partial^2 L_{t}(\theta)}{\partial\theta \partial\theta'} \right\| = o_p(1) $.
    \end{enumerate}
\end{lemma}
\begin{proof}
    This follows directly from Lemma \ref{lemma1}.
\end{proof}

\begin{lemma}\label{lemma4}
    If the conditions of Theorem \ref{theorem:normal} hold, then \( E L_t(\theta) \) has a unique maximum.
\end{lemma}
\begin{proof}
    Firstly, to prove that
    If $ \lambda_1 $ and $ \lambda_2 $ are $ p + q $ and $ 1 + p + q $ dimensional vectors, respectively,
    then for any given $ i $,
    \begin{equation}\label{lemma4:M1}
        \lambda'_1\bz_{1it}\asequal0 \text{ implies } \lambda_1=0
    \end{equation}
    and
    \begin{equation}\label{lemma4:M2}
        \lambda'_2\bz_{2it}\asequal0 \text{ implies } \lambda_2=0
    \end{equation}
    hold.
    In fact, for any given $ i $,
    expand \eqref{lemma4:M1} to
    \begin{equation*}
        \begin{aligned}
            0 \asequal & ~ \lambda_{11}n_i^{-1}\sum_{j\neq i}a_{ij}y_{j(t-1)} + \cdots + \lambda_{1p}n_i^{-1}\sum_{j\neq i}a_{ij}y_{j(t-p)} \\
                       & + \lambda_{1(p+1)}y_{i(t-1)} + \cdots + \lambda_{1(p+q)}y_{i(t-q)}                                                       \\
        \end{aligned}
    \end{equation*}
    If \( \lambda_{11} \neq 0 \), by the definition of model \eqref{eq:NDAR_pq}, there is a node \( s \) such that \( a_{is} = 1 \). 
    In the sense of almost surety, \( y_{i(t-1)} \) can be expressed as a linear combination of 
    \[
    \mathcal{Y}_{i(t-1)} \equiv \left\{ y_{j(t-1)}, j \neq i \right\} \cup \left\{ y_{i(t-r)}, r > 1 \right\},
    \]
    which is denoted as \( \varrho_{s(t-1)} \). 
    By Assumption \eqref{assumption:epsilon}, we know that \( \eta_{s(t-1)} \) is independent of \( \mathcal{Y}_{s(t-1)} \), and therefore independent of \( \varrho_{s(t-1)} \). 
    Thus, we have \( E[\eta_{s(t-1)}^2] = E[\eta_{s(t-1)} \varrho_{s(t-1)}] = 0 \). 
    However, this contradicts the assumption \( E[\eta_{s(t-1)}^2] = 1 \). 
    Hence, \( \lambda_{11} = 0 \).
    Next, if \( \lambda_{1(p+1)} \neq 0 \), then \( y_{i(t-1)} \) can be expressed as a linear combination of \( \mathcal{Y}_{i(t-1)} \). 
    This combination is denoted as \( \varrho_{i(t-1)} \). 
    As above, we can show that \( E[\eta_{i(t-1)}^2] = E[\eta_{i(t-1)} \varrho_{i(t-1)}] = 0 \), 
    which also leads to a contradiction \( E[\eta_{i(t-1)}^2] = 1 \). 
    Hence, \( \lambda_{1(p+1)} = 0 \).
    We have shown that the coefficients associated with time point \( t-1 \) must be zero. 
    That is to say, \( (\lambda_{1,1}, \lambda_{1,p+1}) = (0, 0) \).
    
    Following the same proof procedure, for the coefficients associated with time point \( t-r \), 
    we can prove that \( (\lambda_{1,r}, \lambda_{1,p+r}) = (0, 0) \), \( r=2,3,\ldots \) and so on. 
    By gradual derivation, we can prove that \eqref{lemma4:M1} holds.
    Similarly, we can prove that \eqref{lemma4:M2} holds.

    Given
    \begin{equation}\label{lemma4:eq1}
        \begin{aligned}
            EL_t(\theta) &= -\frac{1}{2}\sum_{i=1}^{N}E\left( \log_{it}(\sigma) + \frac{\varepsilon_{it}^{2}(\mu)}{h_{it}(\sigma)} \right) \\
            &= -\frac{1}{2}\sum_{i=1}^{N}\left( E\log h_{it}(\sigma) + E\frac{h_{it}(\sigma_0)}{h_{it}(\sigma)} \right) - \frac{1}{2}\sum_{i=1}^{N}E\frac{\left[ (\mu_0-\mu)'\bz_{1it} \right]^2}{h_{it}(\sigma)}
        \end{aligned}
    \end{equation}
    Obviously, The second term in \eqref{lemma4:eq1} reaches its maximum value if and only if 
    $ (\mu_0-\mu)'\bz_{1it}=0$ holds almost surely for each $i=1,\cdots,N $.
    The first term in \eqref{lemma4:eq1} is transformed into
    \begin{equation}
        -\frac{1}{2}\sum_{i=1}^{N}\left( -E\log\frac{h_{it}(\sigma_0)}{h_{it}(\sigma)} + E\frac{h_{it}(\sigma_0)}{h_{it}(\sigma)} \right) - \frac{1}{2}\sum_{i=1}^{N}E\log h_{it}(\sigma_0)
    \end{equation}
    It can be seen as a function of \( h_{it}(\sigma_0)/h_{it}(\sigma) \), 
    which reaches its maximum value if and only if $ h_{it}(\sigma_0)/h_{it}(\sigma)=1 $ for every $i=1,\cdots,N $.
    That is, $ (\sigma_0-\sigma)'\bz_{2it}=0 $ is almost surely to hold.
    By \eqref{lemma4:M1} and \eqref{lemma4:M2}, $ \theta_{0} $ is the only maximum point of $EL_t(\theta)$.
\end{proof}

\begin{lemma}\label{lemma5}
    If the conditions of Theorem \ref{theorem:normal} hold, then
    \begin{enumerate}[label= (\roman*), font=\normalfont]
        \item $\Omega_{0}$ is finite and positive definite (a.s.);
        \item $\displaystyle
                  \frac{1}{\sqrt{T}}\sum_{t=1}^{T}
                  \frac{\partial L_t(\theta_0)}{\partial \theta}
                  \longrightarrow_\mathcal{L} N(0, \Sigma_{0}).
              $
    \end{enumerate}
\end{lemma}

\begin{proof}
    (i)
    From the proof of \eqref{lemma2}, the finiteness of $ \Omega_{0} $ is obvious.
    For any $ \lambda=(\lambda'_1, \lambda'_2)' \in \mathbb{R}^{p+q+1}$, 
    \begin{equation*}
        \lambda'\Omega\lambda = E \left[ \sum_{i=1}^{N} \left\{ \lambda'_{1} \frac{\bz_{1it}\bz'_{1it}}{h_{it}(\sigma_0)} \lambda_1 + \lambda'_{2} \frac{\bz_{2it}\bz'_{2it}}{2h_{it}(\sigma_0)} \lambda_2 \right\} \right] \geq 0
    \end{equation*}
    Given \eqref{lemma4:M1} and \eqref{lemma4:M2}, we have
    \begin{equation*}
        \begin{aligned}
            \lambda'\Omega_0\lambda \asequal 0 ~\Leftrightarrow~ & \lambda'_{1} \bz_{1it}\bz'_{1it} \lambda_1 \asequal 0,~ \lambda'_{2} \bz_{2it}\bz'_{2it} \lambda_2 \asequal 0 \\
            ~\Leftrightarrow~                           & \lambda'_{1} \bz_{1it} \asequal 0,~ \lambda'_{2} \bz_{2it}\asequal 0                                       \\
            ~\Leftrightarrow~                           & \lambda_{1}=0,~ \lambda_{2}=0
        \end{aligned}
    \end{equation*}
    Therefore (i) is proved.

    (ii)
    For any $ \lambda\in\mathbb{R}^d $, 
    let $ \mathcal{F}_{t}=\sigma(\bm{y}_{-m}, \cdots, \bm{y}_{t}, \bm{\eta}_{1}, \cdots, \bm{\eta}_{t}) $
    and 
    \begin{equation}
        \lambda'\frac{1}{\sqrt{T}}\sum_{t=1}^{T}\frac{\partial L_t(\theta_0)}{\partial \theta}
        = \sum_{t=1}^{T} \lambda' \frac{1}{\sqrt{T}} \sum_{i=1}^{N} \frac{\partial l_{it}(\theta_0)}{\partial \theta}
        \stackrel{\Delta}{=} \sum_{t=1}^{T} K_{t},
    \end{equation}
    then $ E\left[ K_{t} | \mathcal{F}_{t-1}\right] = 0 $,
    By the Assumption \ref{assumption:y},
    It is easy to prove that
    \begin{equation}
            \sum_{t=1}^{T} E\left[ K_{t}^2 \middle| \mathcal{F}_{t-1}\right] 
            = \sum_{t=1}^{T} \frac{1}{T} \lambda' E \left[ S_{t}(\theta_0)S'_{t}(\theta_0) \middle| \mathcal{F}_{t-1} \right] \lambda
            \to_{p} \lambda' \Sigma_0 \lambda.
    \end{equation}
    For any fixed $ \delta>0 $, we have
    \begin{equation}
        \begin{aligned}
            \sum_{t=1}^{T}E\left[ K_{t}^2 \mathbb{I}_{\{|K_t|\geq \delta\}} \right] &= \frac{1}{T}\sum_{t=1}^{T} E\left[ \lambda'S_t(\theta_0)S'_t(\theta_0)\lambda \mathbb{I}_{\{|\lambda'S_t(\theta_0)S'_t(\theta_0)\lambda|\geq T\delta\}} \right]  \\
            &=E\left[ \lambda'S_t(\theta_0)S'_t(\theta_0)\lambda \mathbb{I}_{\{|\lambda'S_t(\theta_0)S'_t(\theta_0)\lambda|\geq T\delta\}} \right] \stackrel{P}{\longrightarrow} 0 .\\
        \end{aligned}
    \end{equation}
    By the martingale central limit theorem, 
    $ \sum_{t=1}^{T}K_{t} \rightarrow_{\mathcal{L}} N(0, \lambda\Sigma_0\lambda') $.
    Hence by Cramér-Wold device
    \begin{equation}
        \frac{1}{\sqrt{T}}\frac{\partial L_t(\theta_0)}{\partial \theta} \longrightarrow_{\mathcal{L}} N(0,\Sigma_0)
    .\end{equation}
    Formular (ii) is proved.
\end{proof}

\begin{proof}[\textbf{Proof of Theorem \ref{theorem:normal}}]
Now, combine Lemma \ref{lemma2}-\ref{lemma4}, all conditions of Theorem 4.1.1 in Amemiya1985 are satisfied,
then $ \hat{\theta}\to_{\mathcal{P}}\theta_{0} $.
For any sequence \( \tilde{\theta} \) such that \( \tilde{\theta}\to\theta_0 \) in probability, 
by Lemma \ref{lemma3} (iii),
\[
    \frac{1}{T}\sum_{t=1}^{T} \frac{\partial^2 L_t(\tilde{\theta})}{\partial\theta\partial\theta'}  \to_\mathcal{P} \Omega_0.
\]
And with Lemma \ref{lemma5}, all conditions of Theorem 4.1.3 in Amemiya1985 are satisfied,
then Theorem \ref{theorem:normal} is proved.
\end{proof}

\subsection*{Proof of Lemma \ref{lemma:higher_space_estimator}}
In this lemma, 
some of \( \{\phi_r\} \) and \( \{\psi_r\} \) in volatility component 
are allow in the space boundarie, 0.
Then in the proof of Theorem \ref{theorem:normal},
using positive lower bound of theses parameters to ensure the finiteness of the likelihood function,
the score function and the Hessian matrix
are replaced by assumption \( E\left\lVert \by_t \right\rVert^6<\infty \) to ensure.

By the argument of Lemma \ref{lemma2}, we have
\begin{enumerate}[label=(\roman*), font=\normalfont]
    \item $\displaystyle
              E\sup_{\theta\in\Theta}\left\lvert L_{t}(\theta) \right\rvert < \infty ; $
    \item \(\displaystyle E\sup_{\theta\in\Theta} \left\lVert \frac{\partial L_{t}(\theta)}{\partial\theta} \right\rVert < \infty ; \)
    \item $\displaystyle
              E\sup_{\theta\in\Theta} \left\lVert \frac{\partial^2 L_{t}(\theta)}{\partial\theta \partial\theta'} \right\rVert < \infty , $
\end{enumerate}
where the order \( 6 \) of \( E\left\lVert \by_t \right\rVert^6 \) are appear in (iii).
Similarly, Lemma \ref{lemma3}, Lemma \ref{lemma4}, Lemma \ref{lemma5} can be proved.

With the conditions of Theorem \ref{theorem:selector},
$ \hat{\theta}\to_{\mathcal{P}}\theta_{0} $ still holds.
Denote \( \lVert \bm{a} \rVert _B = \bm{a}'B\bm{a} \),
\( F_T(\theta) = T^{-1}\sum_{t=1}^{T}L_t(\theta) \),
\[
    S_T(\theta) = \frac{\partial F_T(\theta)}{\partial \theta},
    \quad
    \Omega_T(\theta) = \frac{\partial^2 F_T(\theta)}{\partial\theta \partial\theta'},
\]
and 
\[
    \Omega_T = \Omega_T(\theta_0), \quad Z_T = -\Omega_T^{-1}T^{1/2}S_T(\theta_0).
\]
By \ref{lemma3} (iii) and \ref{lemma5} (i), 
\( \Omega_0 \) is positive definite and \( \Omega_T - \Omega_0 = o_p(1) \). 
Then, by Lemma 2.3 in Knopov and Korkhin (2011), it is not
hard to show that \( \Omega_T \) is a positive definite matrix (a.s.) for sufficiently large \( T \).
By \ref{lemma5} (ii), we have \( Z_T = O_p(1) \).
Using Taylor's expansion for we get
\begin{equation}
    \begin{aligned}
        F_T(\theta) &= F_T(\theta_0) + S_T(\theta)'(\theta-\theta_0) + \frac{1}{2}(\theta-\theta_0)'\Omega_T(\theta-\theta_0) + R_T(\theta) \\
                    &= F_T(\theta_0) + \frac{1}{2T} \lVert Z_n - \sqrt{T}(\theta-\theta_0) \rVert _{\Omega_T}^{2} - \frac{1}{2T}\left\lVert Z_n \right\rVert _{\Omega_T}^{2} + R_T(\theta) \\
    \end{aligned}
\end{equation}
where \( R_T(\theta) = \frac{1}{2}(\theta-\theta_0)'(J_T - J_T(\tilde{\theta}))(\theta-\theta_0) \), 
and \( \tilde{\theta} \) is between \( \theta \) and \( \theta_0 \).
When \( \theta=\hat{\theta} \), by Lemma \ref{lemma3} (iii) , 
the difference of second-order derivatives tends to zero in probability as \( T \) tends to infinity.
Hence \( R_T(\hat{\theta}) = o_p(\lVert \hat{\theta}-\theta_0 \rVert _{J_T}^{2}) \).
We then have
\begin{equation}\label{eq:F-F}
    \begin{aligned}
        F_T(\hat\theta) - F_T(\theta_0) & = \frac{1}{2T} \lVert Z_n - \sqrt{T}(\hat\theta-\theta_0) \rVert _{\Omega_T}^{2} - \frac{1}{2T}\left\lVert Z_n \right\rVert _{\Omega_T}^{2} + R_T(\hat\theta)                                     \\
                                        & = \frac{1}{2T} \left\{ \lVert Z_n - \sqrt{T}(\hat\theta-\theta_0) \rVert _{\Omega_T}^{2} - \left\lVert Z_n \right\rVert _{\Omega_T}^{2} + o_p(\lVert \sqrt{T}(\hat{\theta}-\theta_0) \rVert _{J_T}^{2}) \right\} \\
                                        & \leq 0
    \end{aligned}
\end{equation}
because \( \hat{\theta} \) minimizes \( F_T(\theta) \).
It follows that
\begin{equation}\label{ieq:Z22}
  \begin{aligned}
    \lVert Z_n - \sqrt{T}(\hat\theta-\theta_0) \rVert _{\Omega_T}^2 & \leq \lVert Z_n \rVert _{\Omega_T}^2 + o_p(\lVert \sqrt{T}(\hat{\theta}-\theta_0) \rVert _{J_T}^{2}) \\
    & \leq \left\{ \lVert Z_n \rVert _{\Omega_T} + o_p(\lVert \sqrt{T}(\hat{\theta}-\theta_0) \rVert _{J_T}) \right\}^{2}
  \end{aligned}
\end{equation}
since \( Z_T = O_p(1) \). By the triangle inequality and \eqref{ieq:Z22} we deduce
\[
    \begin{aligned}
        \lVert \sqrt{T}(\hat\theta-\theta_0) \rVert _{\Omega_T} & \leq \lVert \sqrt{T}(\hat\theta-\theta_0) - Z_T \rVert _{\Omega_T} + \lVert Z_T \rVert _{\Omega_T} \\
                                                                & \leq 2\lVert Z_T \rVert _{\Omega_T} +  o_p(\lVert \sqrt{T}(\hat{\theta}-\theta_0) \rVert _{J_T})
    \end{aligned}
\]
Thus \( \lVert \sqrt{T}(\hat\theta-\theta_0) \rVert _{\Omega_T}(1+o_p(1)) \leq 2\lVert Z_T \rVert _{\Omega_T} = O_p(1) \),
namely, \( \lVert \sqrt{T}(\hat\theta-\theta_0) \rVert = O_p(1) \).
Thus, From \eqref{ieq:Z22}, we have
\( \lVert Z_n - \sqrt{T}(\hat\theta-\theta_0) \rVert _{\Omega_T} = O_p(1) \).
Therefore, by the 2nd line in \eqref{eq:F-F}, 
we have \( F_T(\hat\theta) - F_T(\theta_0) = O_p(1/T) \),
so \( \sum_{t=1}^{T}L_t(\hat\theta) - \sum_{t=1}^{T}L_t(\hat\theta_0) = T(F_T(\hat\theta) - F_T(\theta_0)) = O_p(1) \).

\subsection*{Proof of Theorem~\ref{theorem:selector}}
By Lemma \ref{lemma:higher_space_estimator}
and similar argument in the proof of Theorem 3 in \cite{lin2024vector},
we can prove this theorem.

\end{document}